\documentclass[final,12pt]{article}
\usepackage[english]{babel}
\usepackage{amsthm}
\usepackage{amsmath,amssymb,amsfonts}
\usepackage[dvipsnames,usenames]{color}

\usepackage[bitstream-charter]{mathdesign}
\usepackage[utf8]{inputenc}
\usepackage{tikz}
\usepackage{caption}
\usepackage{subcaption}
\usepackage{enumerate}
\usepackage{amsrefs}
\usepackage{geometry}
\usepackage[affil-it]{authblk}
\newtheorem{thm}{Theorem}[section]
\newtheorem{prop}[thm]{Proposition}
\newtheorem{lem}[thm]{Lemma}
\newtheorem{defi}[thm]{Definition}
\newtheorem{cor}[thm]{Corollary}
\newtheorem{claim}[thm]{Claim}
\newtheorem{conj}{Conjecture}

\def\B{\,\square \,}

\def\N{\mathbb N}

\begin{document}
\title{Classification of vertex-transitive cubic partial cubes}
\author{Tilen Marc\thanks{Electronic address: \texttt{tilen.marc@imfm.si}}}
\affil{Institute of Mathematics, Physics, and Mechanics, Jadranska 19, 1000 Ljubljana, Slovenia}

\maketitle

\begin{abstract}
Partial cubes are graphs isometrically embeddable into hypercubes. In this paper it is proved that every cubic, vertex-transitive partial cube is isomorphic to one of the following graphs: $K_2 \B C_{2n}$, for some $n\geq 2$, the generalized Petersen graph $G(10,3)$, the cubic permutahedron, the truncated cuboctahedron, or the truncated icosidodecahedron. This  classification is a generalization of results of Brešar et ~al.~from 2004 on cubic mirror graphs, it includes all cubic, distance-regular partial cubes (Weichsel, 1992), and presents a contribution to the classification of all cubic partial cubes.
\end{abstract}


\sloppy

\section{Introduction}
Hypercubes are considered to be one of the classic examples of graphs that posses many symmetries. It is a fundamental question to ask how those symmetries are preserved on their subgraphs. To our knowledge the first ones who addressed this question were Brouwer, Dejter and Thomassen in 1992 in \cite{brouwer1993highly}. They provided many surprising and diverse examples of vertex-transitive subgraphs of hypercubes, but did not make a classification. 
Based on their results, examples are very diverse hence a classification seems too ambitious. They suggested that one of the reasons for the latter is that the group of symmetries of a subgraph of a hypercube need not be induced by the group of symmetries of the hypercube. 

On the other hand, Weichsel in 1992 \cite{weichsel1992distance} considered distance-regular subgraphs of hypercubes.
He derived certain properties of them, and  noticed that all his examples are not just subgraphs, but isometric subgraphs of hypercubes. It was thus a natural decision to focus on the symmetries of partial cubes.
He classified all distance-regular partial cubes based on their girth: hypercubes are the only ones with girth four, the six cycle and the middle level graphs are the only ones with girth six, and even cycles of length at least eight are the only ones with higher girths. Notice that all these graphs are vertex-transitive, therefore they are a subfamily of vertex-transitive subgraphs of hypercubes.

Probably the most well-known and studied subfamily of partial cubes are median graphs. It is a well-known result from \cite{mulder1980n} that hypercubes are the only regular -- and thus the only vertex-transitive -- median graphs (for infinite vertex-transitive median graphs check \cite{marc2015vertex}). Due to this result, an extensive study of regular partial cubes has been performed \cite{bonnington2003cubic,brevsar2004cubic,deza2004scale, eppstein2006cubic, klavvzar2003partial, klavvzar2007tribes}. It especially focuses on the cubic case, since the variety of these graphs is far richer than in the case of median graphs. Connections with other geometric structures are established, for example with platonic surfaces \cite{brevsar2004cubic} and simplicial arrangements \cite{eppstein2006cubic}. 

From the point of view of vertex-transitive partial cubes, the most interesting one is the study \cite{brevsar2004cubic}, where a new family of graphs called mirror graphs was introduced. 
Moreover, it was proved that mirror graphs are a subfamily of vertex-transitive partial cubes, and  all mirror graphs that are obtained by cubic inflation (thus are cubic graphs) were classified. In \cite{jaz}, we made an analysis of isometric cycles in a partial cube. In particular, the results imply that there are no cubic partial cubes with girth more than six. This suggests that, as in the case of Weichsel's distance-regular partial cubes, also cubic, vertex-transitive partial cubes should be approached from the study of their girth. In addition, every automorphism of a partial cube $G$ preserves the so-called $\Theta$-classes of $G$, therefore every symmetry of $G$ is induced by a symmetry of a hypercube.

In the present paper we form a natural connection between the study of vertex-transitive subgraphs of hypercubes and the study of cubic partial cubes: we classify all cubic, vertex-transitive partial cubes. The results can be seen as a generalization of results from \cite{brevsar2004cubic}, since all mirror graphs are vertex-transitive partial cubes, and in the cubic case results from \cite{weichsel1992distance}, since every distance-regular partial cube is vertex-transitive.

Let $K_2$ denote the complete graph of order 2, $C_k$ the cycle of length $k$, and $G(n,k)$ the generalized Petersen graph with parameters $3\leq n, 1\leq k  < n/2$. The main result of this paper is the following:

\begin{thm}\label{thm:thethm}
If $G$ is a finite, cubic, vertex-transitive partial cube, then $G$ is isomorphic to one of the following: $K_2 \B C_{2n}$, for some $n\geq 2$, $G(10,3)$, the cubic permutahedron, the truncated cuboctahedron, or the truncated icosidodecahedron.
\end{thm}
\begin{figure}[h]
\centering
\begin{subfigure}[b]{0.3\textwidth}
\begin{tikzpicture}[scale=1.2]
\tikzstyle{vertex}=[draw, circle,fill=white, inner sep=0pt, minimum size =5pt]        

\foreach \a in {0, 36,...,359}
      \draw (\a:1) -- (\a:2);
\foreach \a in {0, 36,...,359}
      \draw (\a+36:2) -- (\a:2);
 \foreach \a in {0, 36,...,359}
      \draw (\a+108:1) -- (\a:1);           
\foreach \a in {0, 36,...,359}
      \node[vertex] at (\a:2) {};
\foreach \a in {0, 36,...,359}
      \node[vertex] at (\a:1) {};

\end{tikzpicture}
\caption{$G(10,3)$}
\end{subfigure}
\qquad
\begin{subfigure}[b]{0.3\textwidth}
\begin{tikzpicture}[scale=1]
\tikzstyle{vertex}=[draw, circle, inner sep=0pt, minimum size =5pt]

\node (0) at (0,1,2) [vertex] {};
\node (1) at (0,2,1) [vertex] {};
\node (2) at (1,0,2) [vertex] {};
\node (3) at (1,2,0) [vertex] {};
\node (4) at (2,1,0) [vertex] {};
\node (5) at (2,0,1) [vertex] {};

\node (6) at (0,-1,2) [vertex] {};
\node (7) at (0,2,-1) [vertex] {};
\node (8) at (-1,0,2) [vertex] {};
\node (9) at (-1,2,0) [vertex] {};
\node (10) at (2,-1,0) [vertex] {};
\node (11) at (2,0,-1) [vertex] {};

\node (12) at (0,1,-2) [vertex] {};
\node (13) at (0,-2,1) [vertex] {};
\node (14) at (1,0,-2) [vertex] {};
\node (15) at (1,-2,0) [vertex] {};
\node (16) at (-2,1,0) [vertex] {};
\node (17) at (-2,0,1) [vertex] {};

\node (18) at (0,-1,-2) [vertex] {};
\node (19) at (0,-2,-1) [vertex] {};
\node (20) at (-1,0,-2) [vertex] {};
\node (21) at (-1,-2,0) [vertex] {};
\node (22) at (-2,-1,0) [vertex] {};
\node (23) at (-2,0,-1) [vertex] {};

\draw (1) -- (3); 
\draw (1) -- (9);
\draw (7) -- (3); 
\draw (7) -- (9);

\draw (0) -- (2); 
\draw (8) -- (6);
\draw (0) -- (8); 
\draw (6) -- (2);
 
\draw (4) -- (5); 
\draw (10) -- (11);
\draw (4) -- (11); 
\draw (5) -- (10);

\draw[dotted] (12) -- (14); 
\draw[dotted] (18) -- (20);
\draw[dotted] (12) -- (20); 
\draw[dotted] (18) -- (14);

\draw (13) -- (15); 
\draw[dotted] (21) -- (19);
\draw (13) -- (21); 
\draw[dotted] (15) -- (19);

\draw (16) -- (17); 
\draw[dotted] (22) -- (23);
\draw[dotted] (16) -- (23); 
\draw (17) -- (22);

\draw (0) -- (1);
\draw (6) -- (13); 
\draw[dotted] (7) -- (12);
\draw[dotted] (18) -- (19); 

\draw (2) -- (5);
\draw (8) -- (17); 
\draw[dotted] (14) -- (11);
\draw[dotted] (20) -- (23); 

\draw (3) -- (4);
\draw (9) -- (16); 
\draw (15) -- (10);
\draw (21) -- (22); 
\end{tikzpicture}
\caption{Cubic permutahedron}
\end{subfigure}
%

\vspace{10pt}
\begin{subfigure}[b]{0.3\textwidth}
\includegraphics[scale=1.4]{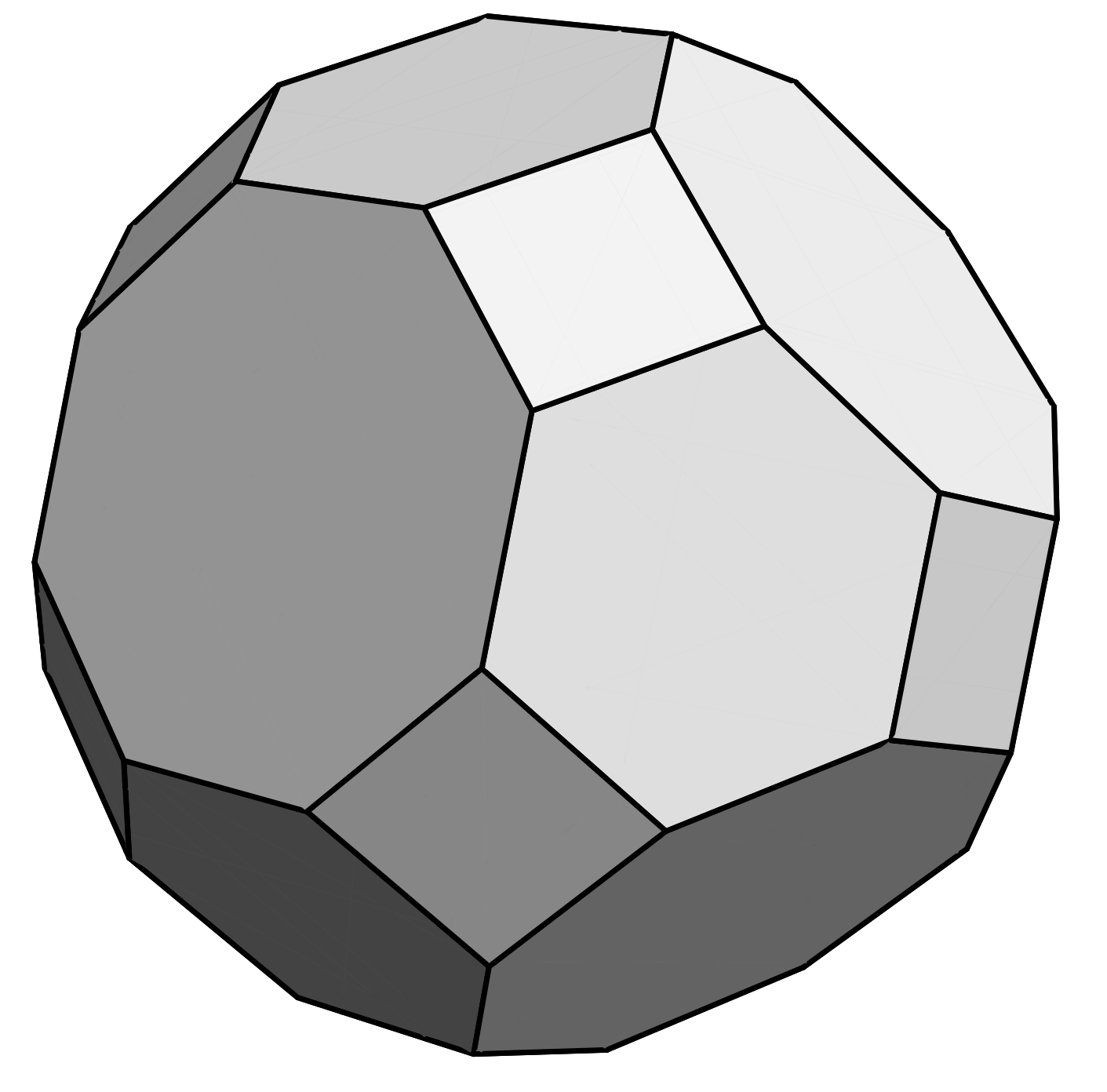}
\caption{Truncated cuboctahedron}
\end{subfigure}
\qquad
\begin{subfigure}[b]{0.4\textwidth}
\includegraphics[scale=0.6]{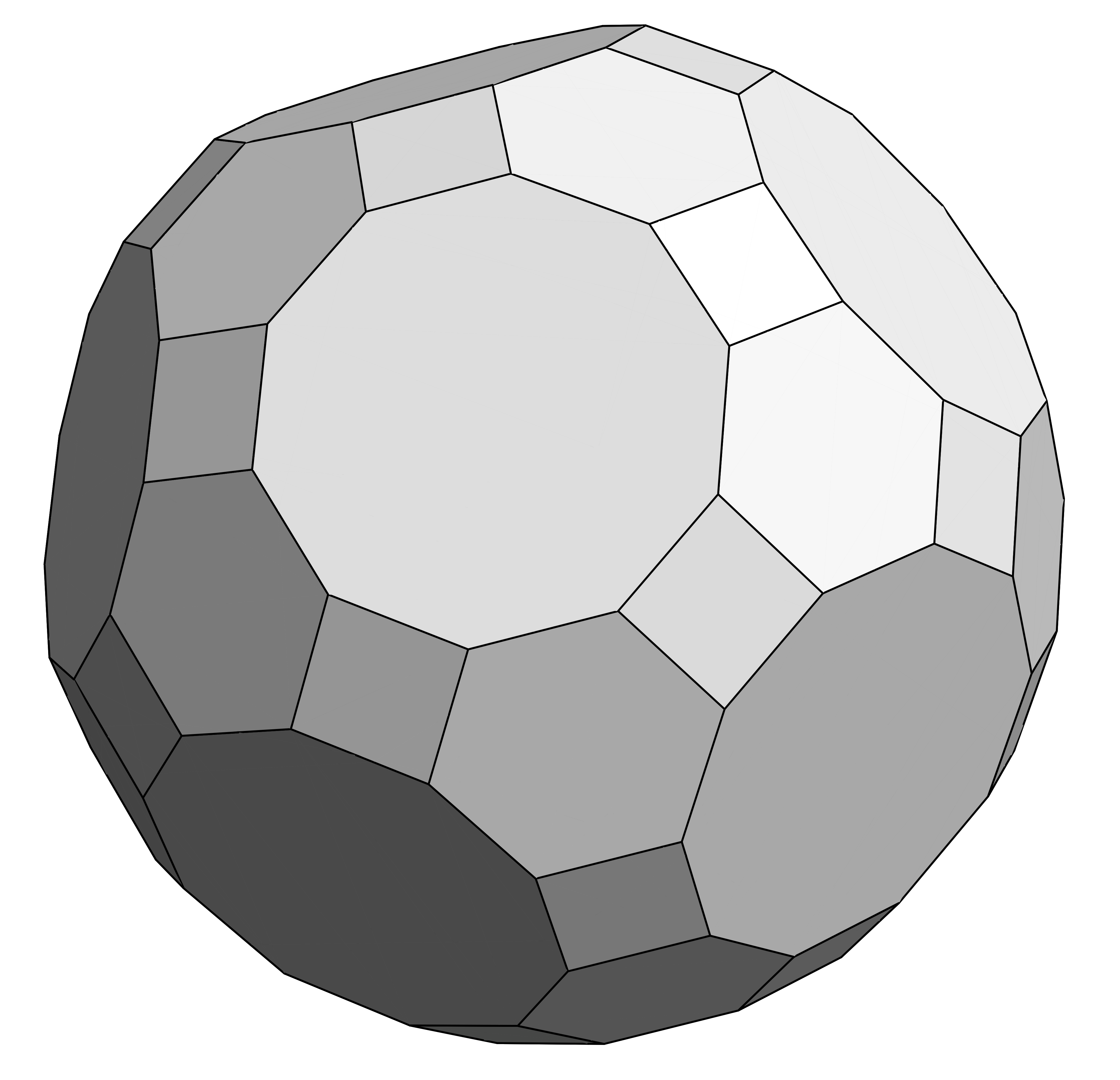}
\caption{Truncated icosidodecahedron}
\end{subfigure}
\caption{The four sporadic examples of cubic, vertex-transitive partial cubes}
\label{fig:finalgraphs2}
\end{figure}
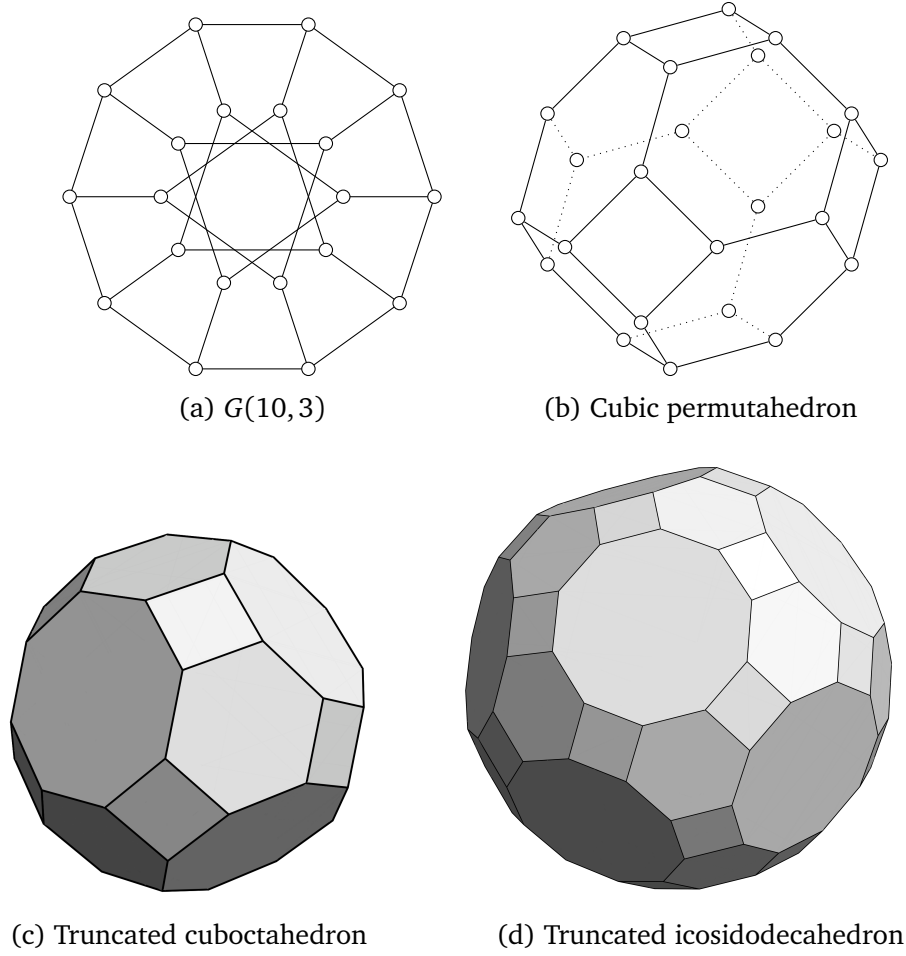

To our surprise, the variety of the graphs from Theorem \ref{thm:thethm} (cf.~Figure \ref{fig:finalgraphs2}) is small, and all graphs are classical graphs that were studied from many (especially geometric) views. We point out that the cubic permutahedron, the truncated cuboctahedron, and the truncated icosidodecahedron are cubic inflations of graphs of platonic surfaces \cite{brevsar2004cubic}, $K_2 \B C_{2n}$ are the only cubic Cartesian products of (vertex-transitive) partial cubes (this includes also the hypercube $Q_3\cong K_2 \B C_4$), while $G(10,3)$ is the only known non-planar cubic partial cube and is isomorphic to the middle level graph of valence three \cite{klavvzar2003partial}.

\section{Preliminaries}
The paper is organized as follows. In this section we briefly present the definitions and results needed to prove Theorem \ref{thm:thethm}, while in the next section we give a proof of it.

We will consider only simple (finite) graphs in this paper. The \emph{Cartesian product} $G\, \square \, H$ of graphs $G$ and $H$ is the graph with the vertex set $V(G) \times V(H)$ and the edge set consists of all pairs $\{(g_1,h_1),(g_2,h_2)\}$ of vertices with $\{g_1,g_2\} \in E(G)$ and $h_1=h_2$, or $g_1=g_2$ and $\{h_1,h_2\} \in E(H)$. \emph{Hypercubes} or \emph{$n$-cubes} are the Cartesian products of $n$-copies of $K_2$. We say a subgraph $H$ of $G$ is \emph{isometric} if for every pair of vertices in $H$ also some shortest path in $G$ connecting them is in $H$. It is \emph{convex} if for every pair of vertices in $H$ all shortest path in $G$ connecting them are in $H$. A \emph{partial cube} is a graph that is isomorphic to an isometric subgraph of some hypercube. The \emph{middle level graph} of valency $n\geq 1$ is the induced subgraph of a hypercube of dimension $2n-1$ on the vertices that have precisely $n$ or $n-1$  coordinates equal to 1, where the coordinates correspond to the factors of the Cartesian product of copies of $K_2$ on vertices $\{0,1\}$.

For a graph $G$, we define the relation $\Theta$ on the edges of $G$ as follows: $ab \Theta xy$ if $d(a,x) + d(b,y)\neq d(a,y) + d(b,x)$, where $d$ is the shortest path distance function. In partial cubes $\Theta$ is an equivalence relation \cite{winkler1984isometric}, and we write $F_{uv}$ for the set of all edges that are in relation $\Theta$ with $uv$. We define $W_{uv}$ as the subgraph induced by all vertices that are closer to vertex $u$ than to $v$, that is $W_{uv}=\langle \{ w: \   d(u,w) < d(v, w) \} \rangle $. In any partial cube $G$, the sets $V(W_{uv})$ and $V(W_{vu})$ partition $V(G)$, with  $F_{uv}$ being the set of edges joining them. We define $U_{uv}$ to be the subgraph induced by the set of vertices in $W_{uv}$ which have a neighbor in $W_{vu}$. For  details and further results, see \cite{Hammack:2011a}.

We shall need a few simple results about partial cubes. All partial cubes are bipartite, since hypercubes are. If $u_1v_1 \Theta u_2v_2$ with $u_2 \in U_{u_1v_1}$, then $d(u_1,u_2)=d(v_1,v_2)$. A path $P$ of a partial cube is  a shortest path or a \emph{geodesic} if and only if  all of its edges belong to pairwise different $\Theta$-classes. If $C$ is a closed walk passing edge $uv$, then $C$ passes an edges in $F_{uv}$ at least two times. By so called Convexity lemma \cite{Hammack:2011a}, convex subgraphs in partial cubes can be characterized as induced, connected subgraphs such that no edge with exactly one end in the subgraph is in relation $\Theta$ with any edge in the subgraph. For the details, we again refer to \cite{Hammack:2011a}.

An \emph{automorphism} of a graph $G$ is a permutation of $V(G)$ that preserves the adjacency of vertices. Graph $G$ is \emph{vertex-transitive} if for every pair $u,v\in V(G)$ there exists an automorphism of $G$ that maps $u$ to $v$. A special subfamily of vertex-transitive graphs comes from groups: For a group $A$ with generator set $S$, such that $S^{-1}=S$ and $\textrm{id}\notin S$,  the \emph{Cayley graph} $\textrm{Cay}(A,S)$ is a graph with vertex set $A$, elements $\alpha_1,\alpha_2$ being adjacent if and only if $\alpha_1\alpha_2^{-1}\in S$. The stabilizer of a vertex $v$ in $G$ is the subgroup of all the automorphisms of $G$ that map $v$ to $v$.
By \cite{sabidussi1958class}, if the stabilizers of a vertex-transitive graph $G$ are trivial, then $G$ is a Cayley graph.



A major part of this paper depends on results developed in \cite{jaz}. The following definition was introduced to study isometric cycles in partial cubes:

\begin{defi}
Let $v_1u_1\Theta v_2u_2$ in a partial cube $G$, with $v_2 \in U_{v_1u_1}$. Let $D_1,\ldots, D_n$ be a sequence of isometric cycles such that $v_1u_1$ lies only on $D_1$, $v_2u_2$ lies only on $D_n$, and each pair $D_i$ and $D_{i+1}$, for $i\in \{1,\ldots,n-1\}$, intersects in exactly one edge from $F_{v_1u_1}$, all the other pairs do not intersect. If the shortest path from $v_1$ to $v_2$ on the union of $D_1,\ldots, D_n$ is isometric in $G$, then we call $D_1,\ldots, D_n$ a \emph{traverse} from $v_1u_1$ to $v_2u_2$. If all the cycles $D_1,\ldots, D_n$ are convex, we call it a \emph{convex traverse}.
\end{defi}

If $D_1,\ldots, D_n$ is a traverse from $v_1u_1$ to $v_2u_2$, then also the shortest path from $u_1$ to $u_2$ on the union of $D_1,\ldots, D_n$ is isometric in $G$. We will call this $u_1,u_2$-shortest path the \emph{$u_1,u_2$-side of the traverse}  and, similarly, the shortest $v_1,v_2$-path on the union of $D_1,\ldots, D_n$ the \emph{$v_1,v_2$-side of the traverse}. The length of these two shortest paths is the \emph{length of the traverse}. It is not difficult to prove the following useful result.

\begin{lem}[\cite{jaz}]\label{lem:cycles}
Let $v_1u_1\Theta v_2u_2$ in a partial cube $G$. Then there exists a convex traverse from $v_1u_1$ to $v_2u_2$.
\end{lem}

We shall also need the following definition:

\begin{defi}
Let $D_1=(v_0v_1\ldots v_mv_{m+1}\ldots v_{2m+2n_1-1})$ and $D_2=(u_0u_1\ldots u_mu_{m+1}\ldots u_{2m+2n_2-1})$ be isometric cycles with $u_0=v_0,\ldots, u_m=v_m$ for $m\geq 2$, and all other vertices pairwise different. We say that $D_1$ and $D_2$ \emph{intertwine} and define $i(D_1,D_2)=n_1+n_2\geq 0$ as the \emph{residue of intertwining}.
\end{defi}

We can calculate the residue of intertwining as $i(D_1,D_2)=(l_1+l_2-4m)/2$, where $l_1$ is the length of $D_1$, $l_2$ the length of $D_2$, and $m$ the number of edges in the intersection. To finish this section we show two simple, but useful properties of convex cycles.

\begin{claim}\label{clm:convex_always_intert}
If two convex cycles share more than an edge or a vertex, then they intertwine.
\end{claim}
\begin{proof}
Suppose that two distinct convex cycles $D_1,D_2$ share two non-adjacent vertices. Let $v_1,v_2 \in V(D_1)\cap V(D_2)$ with maximal distance between them. Since $D_1, D_2$ are convex and distinct, there exists at most one shortest $v_1,v_2$-path and it must be in $V(D_1)\cap V(D_2)$. If any other vertex $u$ is in $V(D_1)\cap V(D_2)$, then also shortest $u,v_1$-, $u,v_2$-paths must be in $V(D_1)\cap V(D_2)$, contradicting the choice of $v_1,v_2$. Thus $D_1$ and $D_2$ share exactly the shortest $v_1,v_2$ path; by the definition they intertwine.
\end{proof}

In \cite{jaz} a similar statement was proved: if at least two isometric cycles in a partial cube share more than an edge or a vertex, then there must be at least two isometric cycles that intertwine.

\begin{claim}\label{clm:4cycle_no_intert}
Every 4-cycle in a partial cube is convex and can share at most an edge or a vertex with any other convex cycle.
\end{claim}
\begin{proof}
Since hypercubes are bipartite and have no induced $K_{2,3}$ (the complete bipartite graph with the bipartition into two and three vertices), the same holds for partial cubes. Therefore every 4-cycle is convex.  Moreover, if a convex cycle $D$ shares more than a vertex or an edge with a 4-cycle, then by Claim \ref{clm:convex_always_intert} it must share exactly two incident edges, say $uv_1$ and $uv_2$. This implies that there are two shortest $v_1,v_2$-paths in $D$, thus $D$ is a 4-cycle. A contradiction, since $K_{2,3}$ is not an induced subgraph of a partial cube.
\end{proof}
Note that in a cubic graph two cycles cannot share exactly a vertex. This fact will be used throughout the paper, sometimes possibly not explicitly pointed out.

\section{Proof of Theorem \ref{thm:thethm}}

We start the proof of Theorem \ref{thm:thethm} by analyzing a simple case that strongly determines the structure of a cubic, vertex-transitive partial cube.

\begin{lem}\label{lem:two4cycles}
If asomevertex of a cubic, vertex-transitive partial cube $G$ lies on two 4-cycles, then $G\cong K_2\, \square \, C_{2n}$, for some $n\geq 2$.
\end{lem}

\begin{proof}
Let $v_0$ lie on two 4-cycles. By Claim \ref{clm:4cycle_no_intert}, they share at most an edge. On the other hand, they must share one edge since the graph is cubic. Let $v_0u_0$ be the edge they share, and let $v_{-1}$ and $v_1$ be the other two neighbors of $v_0$. Moreover, let $u_1$ be the common neighbor of $v_1$ and $u_0$, and similarly $u_{-1}$ be the common neighbor of $v_{-1}$ and $u_0$. If $v_{-1},v_0,v_1$ lie on a common 4-cycle, then, by vertex-transitivity, every vertex lies in three 4-cycles, that pairwise intersect in an edge. It is not hard to see that then $G\cong Q_3$. Thus assume  $v_{-1},v_0,v_1$ do not lie in a common 4-cycle. By vertex-transitivity, also $v_1$ lies in two 4-cycles that intersect in an edge. Since $G$ is cubic, the only possibility is that $v_1u_1$ is the shared edge and that there exist vertices $v_2,u_2$ such that $(v_1v_2u_2u_1)$ is a 4-cycle. If $v_2=v_{-1}$, then by the maximum degree limitation $u_2=u_{-1}$, and thus $G\cong C_3\B K_2$, which is not a partial cube. Also, if $v_2=u_{-1}$, then $u_2=v_{-1}$, and $G$ is not a partial cube. Thus, $v_2$ and $u_2$ are new vertices. We can use the same argument for $v_2,u_2$, as we did for $v_1,u_1$, and find vertices $v_3,u_3$ in a 4-cycle $(v_2v_3u_3u_2)$. Again, we have multiple options: $u_3=u_{-1}$ and $v_3=v_{-1}$, $u_3=v_{-1}$ and $v_3=u_{-1}$, or vertices $v_3,u_3$ are different from all before. In the first case $G\cong K_2\, \square \, C_{4}$, and in the second $G$ is not a partial cube. If the third case occurs, we continue inductively:  at some point the induction stops, therefore $G\cong K_2\, \square \, C_{2n}$, for some $n\in \N$.
\end{proof}

Notice that, if we considered also infinite graphs, a slight modification of the proof would show that the only cubic, vertex-transitive partial cubes with a vertex that lies in two 4-cycles are $K_2\, \square \, C_{2n}$, for $n\geq 2$, and $K_2 \B P_{\infty}$, where $P_{\infty}$ is the two-way infinite path.

%

In \cite{jaz} it was proved that there exists no regular partial cube with degree at least three and girth more than six. For the analysis of partial cubes with girth six, a graph $X$ was introduced as shown in Figure \ref{fig:X}. It was proved, that every regular partial cube with the minimum degree at least three and girth six must have an isometric subgraph isomorphic to $X$. We shall analyze this case in the next two lemmas.

\begin{lem}\label{lem:6and4cycle}
Let $G$ be a cubic partial cube. If a 4-cycle and an isometric 6-cycle in $G$ share two edges, there is a vertex that lies in three 4-cycles. If two isometric 6-cycles in $G$ share more than an edge, then they are a part of an isometric subgraph $X$ or $G$ is a hypercube of dimension 3.
\end{lem}

\begin{proof}
Let $(v_0v_1\ldots v_5)$ be  an isometric 6-cycle and assume vertices $v_0$ and $v_2$ have a common neighbor  $u_1$, different from $v_1$. Then $v_0u_1\Theta v_1v_2 \Theta v_5v_4$. Since $v_0$ and $v_5$ are adjacent, also $u_1$ and $v_4$ are, by the definition of relation $\Theta$. Thus vertex $u_1$ lies in three 4-cycles.

Let isometric 6-cycles $D^1,D^2$ share more than an edge. They cannot share three consecutive edges, by the direct consequence of the transitivity of relation $\Theta$. If they share two opposite edges, the transitivity of relation $\Theta$ implies that $G$ is a hypercube. Since $G$ is cubic, the only remaining option is that they share two consecutive edges. Let $D^1=(v_0v_1\ldots v_5)$ and $D^2=(v_0v_1v_2u_3u_4u_5)$. It holds that $u_3u_4\Theta v_1v_0 \Theta v_3v_4$.
Since $d(u_3,v_3)=2$, it holds $d(u_4,v_4)=2$, by definition of $\Theta$. Thus, there is a vertex $x$ adjacent to $u_4$ and $v_4$. We have found an isometric subgraph $X$ in $G$. 
\end{proof}

\begin{figure}[h]

\centering
\begin{tikzpicture}[style=thick,scale=1.4,label distance=0pt,node distance=10pt,rotate=-45]
\tikzstyle{vertex}=[draw, circle, inner sep=0pt, minimum size =5.5pt]
\node (0) at (0.3  ,0.3) [vertex,label={[shift={(-0.4,-0.5)}]$v_7$}] {};
\node (1) at ( 1  ,0) [vertex,label={[shift={(-0.4,-0.5)}]$v_{6}$ }] {};
\node (2) at (1.7  ,0.3) [vertex,label={[shift={(0.4,-0.5)}]$v_5$}]{};
\node (3) at ( 0  ,1) [vertex,label={[shift={(-0.4,-0.5)}]$v_8$}] {};
\node (4) at (0.9,0.9  ) [vertex,label={[shift={(-0.4,-0.5)}]$c_{1}$ }] {};
\node (5) at ( 2,1) [vertex,label={[shift={(0.4,-0.5)}]$v_{4}$ }] {};
\node (6) at (0.3,1.7) [vertex,label={[shift={(0.4,-0.4)}]$v_{1}$ }] {};
\node (7) at ( 1,2) [vertex,label={[shift={(0.4,-0.5)}]$v_2$}] {};
\node (8) at (1.7  , 1.7) [vertex,label={[shift={(0.4,-0.5)}]$v_{3}$ }] {};
\node (9) at ( 1.1  , 1.1) [vertex,label={[shift={(0.4,-0.5)}]$c_{2}$ }] {};

\draw (0) -- (1);
\draw (1) -- (2);
\draw (0) -- (3);
\draw (3) -- (4);
\draw (4) -- (5);

\draw (2) -- (5);
\draw (3) -- (6);
\draw (5) -- (8);
\draw (6) -- (7);

\draw (7) -- (8);

\draw (1) -- (9);
\draw (9) -- (7);

\end{tikzpicture}
\caption{Graph $X$}
\label{fig:X}
\end{figure}

\begin{lem}\label{lem:H}
If the graph $X$ is an isometric subgraph of a cubic, vertex-transitive partial cube $G$, then $G\cong G(10,3)$.
\end{lem}

\begin{proof}
Assume $X$ is an isometric subgraph of $G$, and denote its vertices as in Figure \ref{fig:X}. Firstly, suppose that there is a 4-cycle in $G$. By vertex-transitivity, $v_2$ must be incident with a 4-cycle. But then an isometric 6-cycle and a 4-cycle must share two edges. Lemmas \ref{lem:two4cycles} and \ref{lem:6and4cycle} imply that $G=C_{4}\B K_2$, but then $G$ does not have an isometric subgraph isomorphic to $X$. Thus there are no 4-cycles in $G$.

Notice that $v_2$ in $X$ lies in three isometric 6-cycles. Since $G$ is vertex-transitive, also $c_1$ must lie in at least three isometric 6-cycles. Since $X$ is an isometric subgraph in a cubic graph without 4-cycles, and no two 6-cycles intersect in three edges, there must be a path $P'$ of length 4, connecting $c_1$ with one of the vertices $v_1,v_3,v_5,v_7$. Moreover, $P'$ can intersect $X$ only in its endpoints. Without loss of generality, assume    $P'$ connects $c_1$ and $v_7$ and denote the vertices of $P'$ by $c_1,s_1,u_6,u_7,v_7$, respecting the order in $P'$. Consider the isometric cycle  $D_1=(c_1s_1u_6u_7v_7v_8)$ and the isometric cycle  $D_2=(c_1v_4v_5v_6v_7v_8)$. Cycles $D_1$ and $D_2$ intertwine in two edges, thus, by Lemma \ref{lem:6and4cycle}, there is a vertex $u_5$ adjacent to $v_5$ and $u_6$. By the isometry of $X$, $u_5$ is distinct from all vertices of $X$. Call $X'$ the graph induced by $V(X)$ and $s_1,u_5,u_6,u_7$. Since there is no 4-cycle in $G$ and $X$ is an isometric subgraph, $X'$ is as represented in Figure \ref{fig:X'X''}a.

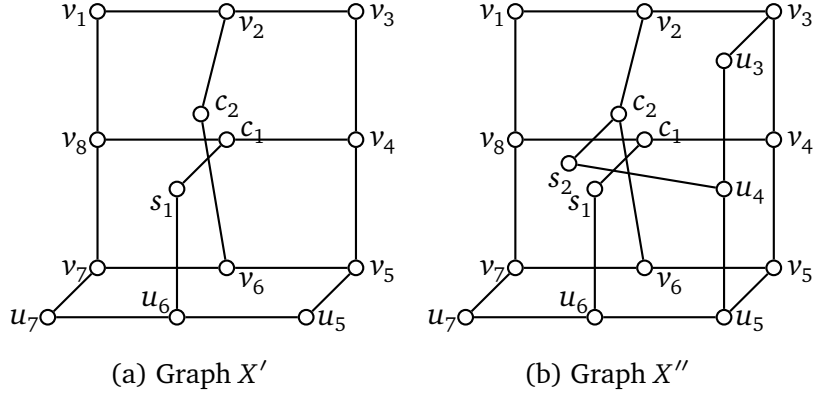
\begin{figure}[h]

\centering
\begin{subfigure}[b]{0.3\textwidth}

\begin{tikzpicture}[style=thick,scale=1.7]
\tikzstyle{vertex}=[draw, circle, inner sep=0pt, minimum size=5.5pt]

\node (0) at (0  ,0,0) [vertex,label={[shift={(-0.3,-0.45)}]$v_{7}$ }] {};
\node (1) at ( 1  ,0,0) [vertex,label={[shift={(0.33,-0.61)}]$v_{6}$ }] {};
\node (2) at (2  ,0,0) [vertex,label={[shift={(0.35,-0.45)}]$v_{5}$ }] {};
\node (3) at ( 0  ,1,0) [vertex,label={[shift={(-0.3,-0.45)}]$v_{8}$ }] {};
\node (4) at (1,1 ,0 ) [vertex,label={[shift={(0.35,-0.3)}]$c_{1}$ }] {};
\node (5) at ( 2,1,0) [vertex,label={[shift={(0.35,-0.45)}]$v_{4}$ }] {};
\node (6) at (0,2,0) [vertex,label={[shift={(-0.3,-0.45)}]$v_{1}$ }] {};
\node (7) at ( 1,2,0) [vertex,label={[shift={(0.33,-0.61)}]$v_{2}$ }] {};
\node (8) at (2  , 2,0) [vertex,label={[shift={(0.35,-0.45)}]$v_{3}$ }] {};
\node (9) at ( 0.8  , 1.2,0) [vertex,label={[shift={(0.35,-0.3)}]$c_{2}$ }] {};

\node (10) at (1,1,1  ) [vertex,label={[shift={(-0.18,-0.63)}]$s_{1}$ }] {};
\node (11) at ( 1  ,0,1) [vertex,label={[shift={(-0.27,-0.2)}]$u_{6}$ }] {};
\node (12) at ( 0,0,1) [vertex,label={[shift={(-0.3,-0.45)}]$u_{7}$ }] {};
\node (13) at (2  ,0,1) [vertex,label={[shift={(0.35,-0.45)}]$u_{5}$ }] {};

\draw (10) -- (4);
\draw (10) -- (11);
\draw (11) -- (12);
\draw (11) -- (13);
\draw (0) -- (12);
\draw (2) -- (13);

\draw (0) -- (1);
\draw (1) -- (2);
\draw (0) -- (3);
\draw (3) -- (4);
\draw (4) -- (5);

\draw (2) -- (5);
\draw (3) -- (6);
\draw (5) -- (8);
\draw (6) -- (7);

\draw (7) -- (8);

\draw (1) -- (9);
\draw (9) -- (7);

\end{tikzpicture}
\caption{Graph $X'$}
\end{subfigure}
\quad
\begin{subfigure}[b]{0.3\textwidth}

\begin{tikzpicture}[style=thick,scale=1.7]
\tikzstyle{vertex}=[draw, circle, inner sep=0pt, minimum size=5.5pt]

\node (0) at (0  ,0,0) [vertex,label={[shift={(-0.3,-0.45)}]$v_{7}$ }] {};
\node (1) at ( 1  ,0,0) [vertex,label={[shift={(0.33,-0.61)}]$v_{6}$ }] {};
\node (2) at (2  ,0,0) [vertex,label={[shift={(0.35,-0.45)}]$v_{5}$ }] {};
\node (3) at ( 0  ,1,0) [vertex,label={[shift={(-0.3,-0.45)}]$v_{8}$ }] {};
\node (4) at (1,1 ,0 ) [vertex,label={[shift={(0.35,-0.3)}]$c_{1}$ }] {};
\node (5) at ( 2,1,0) [vertex,label={[shift={(0.35,-0.45)}]$v_{4}$ }] {};
\node (6) at (0,2,0) [vertex,label={[shift={(-0.3,-0.45)}]$v_{1}$ }] {};
\node (7) at ( 1,2,0) [vertex,label={[shift={(0.33,-0.61)}]$v_{2}$ }] {};
\node (8) at (2  , 2,0) [vertex,label={[shift={(0.35,-0.45)}]$v_{3}$ }] {};
\node (9) at ( 0.8  , 1.2,0) [vertex,label={[shift={(0.35,-0.3)}]$c_{2}$ }] {};

\node (10) at (1,1,1  ) [vertex,label={[shift={(-0.18,-0.63)}]$s_{1}$ }] {};
\node (11) at ( 1  ,0,1) [vertex,label={[shift={(-0.27,-0.2)}]$u_{6}$ }] {};
\node (12) at ( 0,0,1) [vertex,label={[shift={(-0.3,-0.45)}]$u_{7}$ }] {};
\node (13) at (2  ,0,1) [vertex,label={[shift={(0.35,-0.45)}]$u_{5}$ }] {};

\node (14) at ( 0.8 , 1.2,1) [vertex,label={[shift={(-0.1,-0.67)}]$s_{2}$ }] {};
\node (15) at (2  , 2,1) [vertex,label={[shift={(0.35,-0.45)}]$u_{3}$ }] {};
\node (16) at ( 2,1,1) [vertex,label={[shift={(0.35,-0.45)}]$u_{4}$ }] {};

\draw (9) -- (14);
\draw (14) -- (16);
\draw (16) -- (15);
\draw (15) -- (8);
\draw (16) -- (13);

\draw (10) -- (4);
\draw (10) -- (11);
\draw (11) -- (12);
\draw (11) -- (13);
\draw (0) -- (12);
\draw (2) -- (13);

\draw (0) -- (1);
\draw (1) -- (2);
\draw (0) -- (3);
\draw (3) -- (4);
\draw (4) -- (5);

\draw (2) -- (5);
\draw (3) -- (6);
\draw (5) -- (8);
\draw (6) -- (7);

\draw (7) -- (8);

\draw (1) -- (9);
\draw (9) -- (7);

\end{tikzpicture}
\caption{Graph $X''$}
\end{subfigure}

\caption{Induced subgraphs}
\label{fig:X'X''}
\end{figure}

Vertex $c_2$ must lie in at least three isometric 6-cycles. As before, the fact that $X$ is isometric, that the maximal degree of $G$ is three, and the absence of 4-cycles imply that there  exists a path $P''$ of length 4, connecting $c_2$ and one of $v_1,v_3,v_5,v_7$. By symmetry we can limit ourselves to two possibilities. 

If $P''$ connects $c_2$ and $v_5$, then $P''$ must be on $c_2, s_2,u_4,u_5,v_5$, where $s_2$ and $u_4$ are some two vertices in $G$ different from the vertices in $V(X')$ (since $X'$ is an induced graph). Then the cycle $D_3=(c_2 s_2u_4u_5v_5v_6)$ and the cycle $D_4=(c_2v_2v_3v_4v_5v_6)$ meet in edges $v_5v_6$ and $v_6c_2$. By Lemma \ref{lem:6and4cycle}, there must exist a vertex $u_3$ connecting $v_3$ and $u_4$. Since $X'$ is an induced subgraph, $u_3$ is distinct from all the vertices of $X'$. Denote  the graph induced on $V(X')$  and $s_2,u_3,u_4$ by $X''$. Again, since $G$ has no 4-cycles and $X$ is an isometric subgraph, $X''$ is as in Figure \ref{fig:X'X''}b.

On the other hand, $P''$ can connect $c_2$ and $v_3$. For the same reasons as before, there must exist vertices $s_2,u_4,u_3$, different from the vertices of $X'$ such that $P''$ lies on $c_2,s_2,u_4,u_3,v_3$. Then the cycle $D_4$ from above and the cycle $D_5=(c_2s_2u_4u_3v_3v_2)$ share edges $c_2v_2$ and $v_2v_3$. Lemma \ref{lem:6and4cycle} implies that there must be a common neighbor of $v_5$ and $u_4$. This can only be $u_5$, thus we again have $X''$ as an induced subgraph.

\begin{figure}[h]
\centering
\begin{subfigure}[b]{0.3\textwidth}
\begin{tikzpicture}[style=thick,scale=1.7]
\tikzstyle{vertex}=[draw, circle, inner sep=0pt, minimum size=5.5pt]

\node (0) at (0  ,0,0) [vertex,label={[shift={(-0.3,-0.45)}]$v_{7}$ }] {};
\node (1) at ( 1  ,0,0) [vertex,label={[shift={(0.33,-0.61)}]$v_{6}$ }] {};
\node (2) at (2  ,0,0) [vertex,label={[shift={(0.35,-0.45)}]$v_{5}$ }] {};
\node (3) at ( 0  ,1,0) [vertex,label={[shift={(-0.3,-0.45)}]$v_{8}$ }] {};
\node (4) at (1,1 ,0 ) [vertex,label={[shift={(0.35,-0.3)}]$c_{1}$ }] {};
\node (5) at ( 2,1,0) [vertex,label={[shift={(0.35,-0.45)}]$v_{4}$ }] {};
\node (6) at (0,2,0) [vertex,label={[shift={(-0.3,-0.45)}]$v_{1}$ }] {};
\node (7) at ( 1,2,0) [vertex,label={[shift={(0.33,-0.61)}]$v_{2}$ }] {};
\node (8) at (2  , 2,0) [vertex,label={[shift={(0.35,-0.45)}]$v_{3}$ }] {};
\node (9) at ( 0.8  , 1.2,0) [vertex,label={[shift={(0.35,-0.3)}]$c_{2}$ }] {};

\node (10) at (1,1,1  ) [vertex,label={[shift={(-0.18,-0.63)}]$s_{1}$ }] {};
\node (11) at ( 1  ,0,1) [vertex,label={[shift={(-0.27,-0.2)}]$u_{6}$ }] {};
\node (12) at ( 0,0,1) [vertex,label={[shift={(-0.3,-0.45)}]$u_{7}$ }] {};
\node (13) at (2  ,0,1) [vertex,label={[shift={(0.35,-0.45)}]$u_{5}$ }] {};

\node (14) at ( 0.8 , 1.2,1) [vertex,label={[shift={(-0.1,-0.67)}]$s_{2}$ }] {};
\node (15) at (2  , 2,1) [vertex,label={[shift={(0.35,-0.45)}]$u_{3}$ }] {};
\node (16) at ( 2,1,1) [vertex,label={[shift={(0.35,-0.45)}]$u_{4}$ }] {};

\node (17) at ( 1,2,1) [vertex,label={[shift={(-0.27,-0.2)}]$u_{2}$ }] {};
\node (18) at ( 0  ,1,1) [vertex,label={[shift={(-0.3,-0.45)}]$u_{8}$ }] {};

\draw (17) -- (10);
\draw (17) -- (15);
\draw (18) -- (14);
\draw (18) -- (12);

\draw (9) -- (14);
\draw (14) -- (16);
\draw (16) -- (15);
\draw (15) -- (8);
\draw (16) -- (13);

\draw (10) -- (4);
\draw (10) -- (11);
\draw (11) -- (12);
\draw (11) -- (13);
\draw (0) -- (12);
\draw (2) -- (13);

\draw (0) -- (1);
\draw (1) -- (2);
\draw (0) -- (3);
\draw (3) -- (4);
\draw (4) -- (5);

\draw (2) -- (5);
\draw (3) -- (6);
\draw (5) -- (8);
\draw (6) -- (7);

\draw (7) -- (8);

\draw (1) -- (9);
\draw (9) -- (7);

\end{tikzpicture}
\caption{Graph $X'''$}
\end{subfigure}
\qquad
\begin{subfigure}[b]{0.3\textwidth}
\begin{tikzpicture}[style=thick,scale=1.7]
\tikzstyle{vertex}=[draw, circle, inner sep=0pt, minimum size=5.5pt]

\node (0) at (0  ,0,0) [vertex,label={[shift={(-0.3,-0.45)}]$v_{7}$ }] {};
\node (1) at ( 1  ,0,0) [vertex,label={[shift={(0.33,-0.61)}]$v_{6}$ }] {};
\node (2) at (2  ,0,0) [vertex,label={[shift={(0.35,-0.45)}]$v_{5}$ }] {};
\node (3) at ( 0  ,1,0) [vertex,label={[shift={(-0.3,-0.45)}]$v_{8}$ }] {};
\node (4) at (1,1 ,0 ) [vertex,label={[shift={(0.35,-0.3)}]$c_{1}$ }] {};
\node (5) at ( 2,1,0) [vertex,label={[shift={(0.35,-0.45)}]$v_{4}$ }] {};
\node (6) at (0,2,0) [vertex,label={[shift={(-0.3,-0.45)}]$v_{1}$ }] {};
\node (7) at ( 1,2,0) [vertex,label={[shift={(0.33,-0.61)}]$v_{2}$ }] {};
\node (8) at (2  , 2,0) [vertex,label={[shift={(0.35,-0.45)}]$v_{3}$ }] {};
\node (9) at ( 0.8  , 1.2,0) [vertex,label={[shift={(0.35,-0.3)}]$c_{2}$ }] {};

\node (10) at (1,1,1  ) [vertex,label={[shift={(-0.18,-0.63)}]$s_{1}$ }] {};
\node (11) at ( 1  ,0,1) [vertex,label={[shift={(-0.27,-0.2)}]$u_{6}$ }] {};
\node (12) at ( 0,0,1) [vertex,label={[shift={(-0.3,-0.45)}]$u_{7}$ }] {};
\node (13) at (2  ,0,1) [vertex,label={[shift={(0.35,-0.45)}]$u_{5}$ }] {};

\node (14) at ( 0.8 , 1.2,1) [vertex,label={[shift={(-0.1,-0.67)}]$s_{2}$ }] {};
\node (15) at (2  , 2,1) [vertex,label={[shift={(0.35,-0.45)}]$u_{3}$ }] {};
\node (16) at ( 2,1,1) [vertex,label={[shift={(0.35,-0.45)}]$u_{4}$ }] {};

\node (17) at ( 1,2,1) [vertex,label={[shift={(-0.27,-0.2)}]$u_{2}$ }] {};
\node (18) at ( 0  ,1,1) [vertex,label={[shift={(-0.3,-0.45)}]$u_{8}$ }] {};

\node (19) at ( 0  ,2,1) [vertex,label={[shift={(-0.3,-0.45)}]$u_{1}$ }] {};

\draw (19) -- (6);
\draw (18) -- (19);
\draw (17) -- (19);

\draw (17) -- (10);
\draw (17) -- (15);
\draw (18) -- (14);
\draw (18) -- (12);

\draw (9) -- (14);
\draw (14) -- (16);
\draw (16) -- (15);
\draw (15) -- (8);
\draw (16) -- (13);

\draw (10) -- (4);
\draw (10) -- (11);
\draw (11) -- (12);
\draw (11) -- (13);
\draw (0) -- (12);
\draw (2) -- (13);

\draw (0) -- (1);
\draw (1) -- (2);
\draw (0) -- (3);
\draw (3) -- (4);
\draw (4) -- (5);

\draw (2) -- (5);
\draw (3) -- (6);
\draw (5) -- (8);
\draw (6) -- (7);

\draw (7) -- (8);

\draw (1) -- (9);
\draw (9) -- (7);

\end{tikzpicture}
\caption{Graph $G(10,3)$}
\end{subfigure}
\caption{Induced subgraphs}
\label{fig:X'''}
\end{figure}
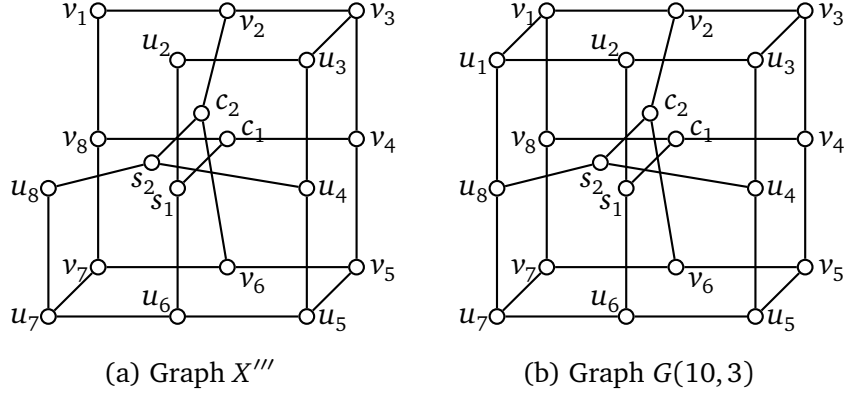

We continue in the same fashion. Cycle $D_3$ and the cycle  $(v_5v_6v_7u_7u_6u_5)$ share $v_6v_5$ and $v_5u_5$, thus there must be a vertex $u_8$ connecting $s_2$ and $u_7$. Similarly, the cycle $(c_1s_1u_6u_5v_5v_4)$ and the cycle $(v_3v_4v_5u_5u_4u_3)$ share $v_4v_5$ and $v_5u_5$ thus there must be a vertex $u_2$ connecting $s_1$ and $u_3$. Let $X'''$ be the subgraph induced   on vertices $V(X'')$ and $u_8,u_2$. Since $X$ is an isometric subgraph and $G$ is without 4-cycles, $X'''$ is as in Figure \ref{fig:X'''}a.

Notice that  the only vertices in $X'''$ that do not have degree 3 are $v_1,u_2,u_8$. Also, observe that $v_5$ lies in six 6-cycles. The only option that $v_1$ lies in six 6-cycles is that there exists $u_1$ connected to  $v_1,u_2$ and $u_8$. It can be checked directly that the obtained graph (shown in Figure \ref{fig:X'''}b) is isomorphic to $G(10,3)$.
\end{proof}


It is a well-known fact, that the edge-connectivity of a vertex-transitive graph equals the degree of its vertices \cite{godsil2001algebraic}. In a  cubic, vertex-transitive partial cubes this implies that $|F_{ab}|\geq 3$ for every edge $ab\in E(G)$.

\begin{lem}\label{lem:allincident}
In a cubic, vertex-transitive partial cube $G$, every pair of incident edges lies in a convex cycle.
\end{lem}

\begin{proof}
If the girth of $G$ is more than 4, then $G$ has an isometric subgraph isomorphic to $X$ and, by Lemma \ref{lem:H}, $G$ is isomorphic to $G(10,3)$. The assertion holds in this graph.

Now assume that the girth of $G$ is 4. Let $D=(v_0v_1v_2v_3)$ be a 4-cycle in $G$. Let $ab$ be an edge in $F_{v_0v_1}$, distinct from $v_0v_1$ and $v_2v_3$. By Lemma \ref{lem:cycles}, there exists a convex traverse from $v_0v_1$ to $ab$. Let $D'$ be the first convex cycle on this traverse. Without loss of generality assume $D'\neq D$ (if otherwise take the second cycle on the traverse and exchange edge $v_0v_1$ with $v_2v_3$). Similarly, without loss of generality $v_1v_2$ lies in a convex cycle $D''$, different from $D$. Since convex cycles $D'$ and $D''$ each share at most an edge with $D$, by Claim \ref{clm:4cycle_no_intert}, all the pairs of edges incident with $v_1$ lie in some convex cycle. By transitivity, this holds for all the vertices.
\end{proof}

Let $u$ be an arbitrary vertex of a cubic, vertex-transitive partial cube $G$, and let $u_1,u_2,u_3$ be its neighbors. Let $g_1(G)$ be the length of a shortest convex cycle on $u_1,u,u_2$, let $g_2(G)$ be the length of a shortest convex cycle on $u_2,u,u_3$, and let $g_3(G)$ be the length of a shortest convex cycle on $u_3,u,u_1$. Without loss of generality assume $g_1(G) \leq g_2(G) \leq g_3(G)$. Clearly, for a vertex-transitive partial cube functions $g_1,g_2,g_3$ are independent of the choice of vertex $u$. Lemmas \ref{lem:two4cycles} and \ref{lem:H}, together with the fact that a cubic partial cube with girth at least 6 includes an isometric subgraph $X$ \cite{jaz}, immediately  give the following corollary. 

\begin{cor}\label{cor:basic}
Let $G$ be a cubic, vertex-transitive partial cube. Then $g_1(G)\leq 6$. If $g_1(G)=6$, then $G\cong G(10,3)$, while if $g_1(G)=g_2(G)=4$, then $G\cong K_2\, \square \, C_{2n}$, for some $n\geq 2$.
\end{cor}

We cover another important case in the next lemma and obtain a well known cubic, vertex-transitive partial cube \cite{gedeonova1990constructions}.

\begin{lem}\label{lem:6-cycle}
Let $G$ be a cubic, vertex-transitive partial cube with $g_1(G)=4,g_2(G)=g_3(G)=6$. Then $G$ is isomorphic to the cubic permutahedron.
\end{lem}

\begin{proof}
We paste a disc on every convex 6-cycle and every 4-cycle. By Claim \ref{clm:4cycle_no_intert}, Lemma \ref{lem:6and4cycle} and the fact that $G$ is not isomorphic to  $G(10,3)$ or $ K_2\, \square \, C_{2n}$, two 4-cycles, two convex 6-cycles, or a 4-cycle and a 6-cycle share at most an edge. Thus we obtain a closed surface.  

Denote by $f_6$ the number of convex 6-cycles, by $f_4$ the number of 4-cycles, by $f$ the number of faces of the embedding of $G$, by $n$ and $e$ the number of vertices and edges of $G$, and by $\chi$ the Euler characteristic of the surface. We have
$$3n=2e, \quad f_4+f_6=f, \quad  4f_4=6f_6/2=n, \textrm{ and } n-e+f=\chi.$$
From the second and third equation we get $n(\frac{1}{4}+\frac{2}{6})=f$. If we use the latter combined with the first equation in the Euler formula, we get:
$$\chi=n\left(1-\frac{3}{2}+\frac{1}{4}+\frac{2}{6}\right)=\frac{n}{12}.$$
Since the right hand side of the equation is positive, it holds that $\chi>0$, i.e.,~$\chi=1$ or $\chi=2$. In both cases $n\leq 24$. Cubic partial cubes up to 32 vertices are known, the only vertex-transitive on 24 vertices are the cubic permutahedron and $K_2\B C_{12}$, while $K_2\B C_6$ is the only one on 12 vertices. Thus $G$ must be isomorphic to the cubic permutahedron.
\end{proof}

For the sake of convenience, we shall call the graphs $C_{2n}\B K_2$ (for $n\geq 2$), $G(10,3)$, and the cubic permutahedron the \emph{basic cubic graphs}. In what follows, we will find all non-basic, cubic, vertex-transitive partial cubes. For this we shall need a simple but technical lemma.

\begin{lem}\label{lem:paste2}
Let $u_0v_0 \Theta u_mv_m$ with $u_m\in U_{u_0v_0}$ in a partial cube $G$. If $P=u_0u_1\ldots u_m$ is a geodesic, then at least one of the following holds (Cases (i)-(iii) are illustrated in Figure \ref{fig:Cases}):

\begin{enumerate}[(i)]
\item There exist vertices $w_{i_1},w_{i_2},\ldots, w_{i_l}\notin V(P)$, for some $l\geq 0$ and $0< i_1<i_2-1,i_2<i_3-1, \ldots, i_{l-1}<i_l-1<m-1$, such that the path $u_0u_1\ldots u_{i_1-1}w_{i_1}u_{i_1+1}\ldots u_{i_l-1}w_{i_l}u_{i_l+1}\ldots u_m$ is the $u_0,u_m$-side of some convex traverse $T$ from $v_0u_0$ to $v_mu_m$.

\item There exist edges $u_{i}z_{i}, u_{j}u_{j+1}$, and vertices $w_{i_1},w_{i_2},\ldots, w_{i_l} \notin V(P)$, for some $l\geq 0$, $0\leq i< i_1<i_2-1,i_2<i_3-1, \ldots, i_{l-1}<i_l-1< j-1 \leq m-1$, such that the path $u_{i}u_{i+1}\ldots u_{i_1-1}w_{i_1}u_{i_1+1}\ldots u_{i_l-1}w_{i_l}u_{i_l+1}\ldots u_{j}$ is the $u_{i},u_{j}$-side of some convex traverse $T$ from $z_{i}u_{i}$ to $u_{j+1}u_{j}$ of length at least two.

\item There exist a vertex $w_i$ adjacent to $u_{i-1}$ and $u_{i+1}$, edges $w_{i}z_{i}$ and $u_{j}u_{j+1}$, and
vertices $w_{i_1},\ldots, w_{i_l}\notin V(P)$, for some $l\geq 0$ and $0< i<i_1-1,i_1<i_2-1, \ldots, i_{l-1}<i_l-1< j-1 \leq m-1$, such that the path $w_{i}u_{i+1}\ldots u_{i_l-1}w_{i_l}u_{i_l+1}\ldots u_{j}$ is the $w_i,u_j$-side of some convex traverse $T$ from $z_{i}w_{i}$ to $u_{j+1}u_{j}$.
\end{enumerate}
\end{lem}

\begin{proof}

\begin{figure}[h]

\centering
\begin{subfigure}[b]{0.3\textwidth}

\begin{tikzpicture}[style=thick,xscale=0.8,yscale=0.8]
\tikzstyle{vertex}=[draw, circle, inner sep=0pt, minimum size=5.5pt]
\tikzstyle{t}=[ultra thick]

\node (0) at (0  ,0) [vertex,label={[shift={(0.35,-0.45)}]$u_{8}$ }] {};
\node (1) at ( 0  ,1) [vertex,label={[shift={(0.35,-0.45)}]$u_{7}$ }] {};
\node (2) at (1  ,2) [vertex,label={[shift={(0.35,-0.45)}]$u_{6}$ }] {};

\node (3) at ( 0  ,3) [vertex,label={[shift={(0.35,-0.45)}]$u_{5}$ }] {};
\node (4) at (0,4 ) [vertex,label={[shift={(0.35,-0.45)}]$u_{4}$ }] {};
\node (5) at ( 0,5) [vertex,label={[shift={(0.35,-0.45)}]$u_{3}$ }] {};

\node (6) at (1,6) [vertex,label={[shift={(0.35,-0.45)}]$u_{2}$ }] {};
\node (7) at ( 0,7) [vertex,label={[shift={(0.35,-0.45)}]$u_{1}$ }] {};
\node (8) at (0  , 8) [vertex,label={[shift={(0.35,-0.45)}]$u_{0}$ }] {};

\node (9) at ( -1  , 2) [vertex,label={[shift={(0.35,-0.45)}]$w_{6}$ }] {};
\node (10) at ( -1  , 6) [vertex,label={[shift={(0.35,-0.45)}]$w_{2}$ }] {};

\node (11) at ( -2  , 0) [vertex,label={[shift={(-0.35,-0.45)}]$v_{8}$ }] {};
\node (12) at ( -2  , 1) [vertex,label={[shift={(-0.30,-0.68)}]}] {};
\node (13) at ( -3  , 2) [vertex,label={[shift={(-0.35,-0.45)}]}] {};
\node (14) at ( -2  , 3) [vertex,label={[shift={(-0.35,-0.45)}]}] {};
\node (15) at ( -2  , 4) [vertex,label={[shift={(-0.35,-0.45)}]}] {};
\node (16) at ( -2  , 5) [vertex,label={[shift={(-0.35,-0.45)}]}] {};
\node (17) at ( -3  , 6) [vertex,label={[shift={(-0.35,-0.45)}]}] {};
\node (18) at ( -2  , 7) [vertex,label={[shift={(-0.35,-0.45)}]}] {};
\node (19) at ( -2  , 8) [vertex,label={[shift={(-0.35,-0.45)}]$v_{0}$ }] {};

\draw (9) -- (13);
\draw (10) -- (17);
\draw (4) -- (15);

\draw[t] (0) -- (1);
\draw[t] (1) -- (2);
\draw[t] (2) -- (3);
\draw[t] (3) -- (4);
\draw[t] (4) -- (5);
\draw[t] (6) -- (5);
\draw[t] (7) -- (6);
\draw[t] (8) -- (7);

\draw (11) -- (12);
\draw (12) -- (13);
\draw (13) -- (14);
\draw (14) -- (15);
\draw (16) -- (15);
\draw (17) -- (16);
\draw (18) -- (17);
\draw (19) -- (18);

\draw (0) -- (11);
\draw (8) -- (19);
\draw (5) -- (10);
\draw (7) -- (10);

\draw (1) -- (9);
\draw (3) -- (9);

\end{tikzpicture}
\caption{An example of Case (i). The thick edges are the edges of $P$, $w_{i_1}=w_2$, $w_{i_2}=w_6$, $l=2$.}
\end{subfigure}
\quad
\begin{subfigure}[b]{0.3\textwidth}
\begin{tikzpicture}[style=thick,xscale=0.8,yscale=0.8]
\tikzstyle{vertex}=[draw, circle, inner sep=0pt, minimum size=5.5pt]
\tikzstyle{t}=[ultra thick]

\node (0) at (0  ,0) [vertex,label={[shift={(0.50,-0.45)}]$u_{i+6}$ }] {};
\node (1) at ( 0  ,1) [vertex,label={[shift={(0.50,-0.45)}]$u_{i+5}$ }] {};
\node (2) at (1  ,2) [vertex,label={[shift={(0.50,-0.45)}]$u_{i+4}$ }] {};

\node (3) at ( 0  ,3) [vertex,label={[shift={(0.50,-0.45)}]$u_{i+3}$ }] {};
\node (4) at (0,4 ) [vertex,label={[shift={(0.50,-0.45)}]$u_{i+2}$ }] {};
\node (5) at ( 0,5) [vertex,label={[shift={(0.50,-0.45)}]$u_{i+1}$ }] {};

\node (6) at (0,6) [vertex,label={[shift={(0.35,-0.45)}]$u_{i}$ }] {};
\node (7) at ( 0,7) [vertex,label={[shift={(0.50,-0.45)}] }] {};

\node (9) at ( -1  , 2) [vertex,label={[shift={(0.50,-0.45)}]$w_{i+4}$ }] {};

\node (11) at ( 0  , -1) [vertex,label={[shift={(0.50,-0.45)}]$u_{i+7}$ }] {};
\node (12) at ( -2  , 1) [vertex,label={[shift={(-0.30,-0.68)}]}] {};
\node (13) at ( -3  , 2) [vertex,label={[shift={(-0.30,-0.68)}]}] {};
\node (14) at ( -2  , 3) [vertex,label={[shift={(-0.30,-0.68)}]}] {};
\node (15) at ( -2  , 4) [vertex,label={[shift={(-0.30,-0.68)}]}] {};
\node (16) at ( -2  , 5) [vertex,label={[shift={(-0.30,-0.68)}]}] {};
\node (17) at ( -2  , 6) [vertex,label={[shift={(-0.35,-0.45)}]$z_i$}] {};

%

\node (22) at ( 0  , -2) [vertex,label={[shift={(-0.30,-0.68)}] }] {};

\draw[t] (11) -- (22);
\draw (9) -- (13);
\draw (6) -- (17);
\draw (4) -- (15);

\draw[t] (0) -- (1);
\draw[t] (1) -- (2);
\draw[t] (2) -- (3);
\draw[t] (3) -- (4);
\draw[t] (4) -- (5);
\draw[t] (6) -- (5);
\draw[t] (7) -- (6);

\draw (11) -- (12);
\draw (12) -- (13);
\draw (13) -- (14);
\draw (14) -- (15);
\draw (16) -- (15);
\draw (17) -- (16);

\draw[t] (0) -- (11);

\draw (1) -- (9);
\draw (3) -- (9);

\end{tikzpicture}
\caption{An example of Case (ii). The thick edges are edges in $P$, $u_j=u_{i+6}$, $w_{i_1}=w_{i+4}$, and $l=1$.}
\end{subfigure}
\quad
\begin{subfigure}[b]{0.3\textwidth}
\begin{tikzpicture}[style=thick,xscale=0.8,yscale=0.8]
\tikzstyle{vertex}=[draw, circle, inner sep=0pt, minimum size=5.5pt]
\tikzstyle{t}=[ultra thick]

\node (0) at (0  ,0) [vertex,label={[shift={(0.50,-0.45)}]$u_{i+6}$ }] {};
\node (1) at ( 0  ,1) [vertex,label={[shift={(0.50,-0.45)}]$u_{i+5}$ }] {};
\node (2) at (1  ,2) [vertex,label={[shift={(0.50,-0.45)}]$u_{i+4}$ }] {};

\node (3) at ( 0  ,3) [vertex,label={[shift={(0.50,-0.45)}]$u_{i+3}$ }] {};
\node (4) at (0,4 ) [vertex,label={[shift={(0.50,-0.45)}]$u_{i+2}$ }] {};
\node (5) at ( 0,5) [vertex,label={[shift={(0.50,-0.45)}]$u_{i+1}$ }] {};

\node (6) at (-1,6) [vertex,label={[shift={(0.50,-0.45)}]$w_{i}$ }] {};
\node (7) at ( 0,7) [vertex,label={[shift={(-0.30,-0.68)}] }] {};

\node (9) at ( -1  , 2) [vertex,label={[shift={(0.50,-0.45)}]$w_{i+4}$ }] {};

\node (11) at ( 0  , -1) [vertex,label={[shift={(0.50,-0.45)}]$u_{i+7}$ }] {};
\node (12) at ( -2  , 1) [vertex,label={[shift={(-0.30,-0.68)}]}] {};
\node (13) at ( -3  , 2) [vertex,label={[shift={(-0.30,-0.68)}]}] {};
\node (14) at ( -2  , 3) [vertex,label={[shift={(-0.30,-0.68)}]}] {};
\node (15) at ( -2  , 4) [vertex,label={[shift={(-0.30,-0.68)}]}] {};
\node (16) at ( -2  , 5) [vertex,label={[shift={(-0.30,-0.68)}]}] {};
\node (17) at ( -2  , 6) [vertex,label={[shift={(-0.35,-0.45)}]$z_i$}] {};

%

\node (22) at ( 0  , -2) [vertex,label={[shift={(-0.30,-0.68)}] }] {};
\node (23) at (1,6) [vertex,label={[shift={(-0.30,-0.68)}] }] {};

\draw[t] (23) -- (5);
\draw[t] (23) -- (7);

\draw[t] (11) -- (22);
\draw (9) -- (13);
\draw (6) -- (17);
\draw (4) -- (15);

\draw[t] (0) -- (1);
\draw[t] (1) -- (2);
\draw[t] (2) -- (3);
\draw[t] (3) -- (4);
\draw[t] (4) -- (5);
\draw (6) -- (5);
\draw (7) -- (6);

\draw (11) -- (12);
\draw (12) -- (13);
\draw (13) -- (14);
\draw (14) -- (15);
\draw (16) -- (15);
\draw (17) -- (16);

\draw[t] (0) -- (11);

\draw (1) -- (9);
\draw (3) -- (9);

\end{tikzpicture}
\caption{An example of Case (iii). The thick edges are edges in $P$, $u_j=u_{i+6}$, $w_{i_1}=w_{i+4}$, and $l=1$.}
\end{subfigure}
\caption{}
\label{fig:Cases}
\end{figure}
Assume the lemma does not hold and let   $v_0u_0,v_mu_m$ be counterexample edges with geodesic $P=u_0u_1\ldots u_m$ that has length as small as possible. By Lemma \ref{lem:cycles}, there exists a traverse from $v_0u_0$ to $v_mu_m$, and let $P_1$ be the $u_0, u_m$-side of it. If $P_1=P$, then Case (i) in the lemma holds, a contradiction. Thus there exists a cycle $D'=(u_{k_1}u_{k_1+1}\ldots u_{k'}z_{k'-1}\ldots z_{k_1+1})$ for some $0\leq k_1<k'-1\leq m-1$, where the path $u_{k_1}z_{k_1+1}\ldots z_{k'+1}u_{k'}$ is a part of the $u_0,u_m$-side of a traverse from $v_0u_0$ to $v_mu_m$. If this cycle is of length 4, take the next one on $P$ of the same form. If all of them are 4-cycles, we have Case (i), a contradiction.

Therefore assume $D'$ is not a 4-cycle. First assume it is convex. Then edges $u_{k_1}z_{k_1+1}$ and $u_{k'}u_{k'-1}$ are in relation $\Theta$, and $D'$ is a convex traverse from $u_{k_1}z_{k_1+1}$ to $u_{k'}u_{k'-1}$ of length at least two. Thus we have Case (ii).

Now assume $D'$ is not convex. Both $u_{k_1}u_{k_1+1}\ldots u_{k'}$ and $u_{k_1}z_{k_1+1}\ldots z_{k'-1}u_{k'}$ are shortest $u_{k_1}u_{k'}$-paths. Two such paths in a hypercube  cross the same $\Theta$-classes, thus the same must hold in a partial cube.
Let $u_{k_2}u_{k_2+1}$, for some $k_1<k_2<k'$, be the edge in $u_{k_1}u_{k_1+1}\ldots u_{k'}$ that is the same $\Theta$-class as the edge ${u_{k_1}z_{k_1+1}}$. Then the path $P'=u_{k_1}\ldots u_{k_2}$ is shorter than $P$, thus the lemma holds for it. If Case (ii), resp. Case (iii), holds for $P'$, then Case (ii), resp. Case (iii), holds for $P$ as well. If Case (i) holds for $P'$ with edges $u_{k_1}z_{k_1+1},u_{k_2}u_{k_2+1}$, and the traverse from $u_{k_1}z_{k_1+1}$ to $u_{k_2}u_{k_2+1}$ is of length at least two, then Case (ii) holds for $P$.
Thus assume the traverse from $u_{k_1}z_{k_1+1}$ to $u_{k_2}u_{k_2+1}$ is of length one, i.e.,~there is a 4-cycle $D''=(u_{k_1}u_{k_1+1}u_{k_1+2}z_{k_1+1})$. 

Since $D'$ is not a 4-cycle, it holds that $D''\neq D'$.
Let $z_{k_1+2}$ be the neighbor of $z_{k_1+1}$ in $D'$, different from $u_{k_1}$. For the same reasons as above, there must exist an edge $u_{k_3}u_{k_3+1}$, for $k_1+2\leq k_3 <k'$, in $u_{k_1}u_{k_1+1}\ldots u_{k'}$ that is the same $\Theta$-class as the edge $z_{k_1+1}z_{k_1+2}$. Again, the isometric path $P''=z_{k_1+1}u_{k_1+2}\ldots u_{k_3}$ is shorter than $P$, thus the lemma holds for $P''$ with edges $z_{k_1+1}z_{k_1+2}$ and $u_{k_3}u_{k_3+1}$.

If there exists a 4-cycle $(z_{k_1+1}w_{k_1+2}u_{k_1+3}u_{k_1+2})$ for some vertex $w_{k_1+2}$, then Case (iii) holds for edges $z_{k_1+1}w_{k_1+2}$ and $u_{k_1+2}u_{k_1+3}$ with the traverse being a 4-cycle. Thus assume  there is no such 4-cycle.

If Case (i) holds for $P''$ with $z_{k_1+1}z_{k_1+2}$ and $u_{k_3}u_{k_3+1}$, then Case (iii) holds for $P$. If Case (ii) holds for $P''$ with $z_{k_1+1}z_{k_1+2}$ and $u_{k_3}u_{k_3+1}$, then Case (ii) or Case (iii) holds for $P$. Finally, if Case (iii) holds for $P''$ with $z_{k_1+1}z_{k_1+2}$ and $u_{k_3}u_{k_3+1}$, then Case (iii) holds for $P$ since there is no 4-cycle of the form $(z_{k_1+1}w_{k_1+2}u_{k_1+3}u_{k_1+2})$.
\end{proof}

As we can see in Figure \ref{fig:Cases}, Lemma \ref{lem:paste2} provides a traverse $T$ that is attached to the path $P$ as in one of the Cases (i)-(iii) with (possibly) some 4-cycles in between $T$ and $P$. We will use this fact in the following way. First we will show that if we can find in a non-basic, cubic, vertex-transitive partial cube $G$ a convex traverse attached to a side of another convex traverse in a nice way, then the high density of convex cycles in this part of $G$ will imply that some cycles must be small. In particular, we will show that $g_2(G)=6$. Second, we will find this situation in $G$ by considering two complementary types of graph: ether we allow convex cycles to intertwine or we do not. Finally, we will conclude the proof by classifying cubic, vertex transitive partial cubes with $g_1(G)=4, g_2(G)=6$ and $g_3(G)>6$, by using certain results from group theory.

\begin{lem}\label{lem:constHs}
Let $G$ be a non-basic, cubic, vertex-transitive partial cube. Let there be a traverse $T_1$ from $u_0x_0$ to $u_mx_m$
with $P_1=u_0\ldots u_m$ being the $u_0,u_m$-side of it. If we have a traverse $T_2$ attached to $P_1$ as in one of the Cases (i)-(iii) from Lemma \ref{lem:paste2}, with additional assumption that the convex cycles on $T_1$ and convex cycles on $T_2$ pairwise share at most an edge, then $g_2(G)=6$.
\end{lem}

\begin{proof}
Graph $G$ is non-basic, therefore $g_1(G)=4$ and $g_2(G)\geq 6$. Notice that  to prove $g_2(G)=6$ it is enough to find one convex 6-cycle. By Claim \ref{clm:4cycle_no_intert}, this 6-cycle cannot share more than an edge with any 4-cycle, thus it gives its contribution to $g_2(G)=6$.

Let $T_1,T_2$ be as in the assertion of the lemma; we adapt the notation from the Lemma \ref{lem:paste2}. If the length of $T_2$ (which is the shortest of the two traverses) is 2, then there exist two incident 4-cycles on $T_1$, which cannot be by Lemma \ref{lem:two4cycles}, or it includes a convex 6-cycle and we are done. If the length is 1, we have Case (i) or (iii) - in both cases there exist two incident 4-cycles, a contradiction. Thus assume the length of both traverses is at least 3. Also notice that if a 4-cycle of the form $(u_{k-1}u_ku_{k+1}w_k)$ is in between $T_1$ and $T_2$, then by Claim \ref{clm:4cycle_no_intert} there must be edges $u_kx_k$ and $w_kv_k$ such that two convex cycles of $T_1$ share $u_kx_k$, and two convex cycles of $T_2$ share $w_kv_k$.

First assume that $T_2$ is attached as in Case (i). The traverse $T_1$ starts in the edge $u_0x_0$, $T_2$ starts in the edge $u_0v_0$ and they both share the edge $u_0u_1$ since $G$ is cubic. The first cycle on $T_1$ and the first cycle on $T_2$ cannot be both 4-cycles since such cycles are not incident, by Lemma \ref{lem:two4cycles}. If one of them starts with a 6-cycle, we are done. Without loss of generality assume that $T_1$ starts with a $(2l+2)$-cycle for $2l+2\geq 8$. By the last statement of the previous paragraph, there cannot be a 4-cycle in between $T_1$ and $T_2$ of the form $(u_{k'-1}u_{k'}u_{k'+1}w_{k'})$ for $0\leq k'<l$. Thus path $u_0u_1\ldots u_{l-1}$ must be the beginning of the sides of $T_1$ and $T_2$. Since any two cycles of $T_1$ and $T_2$ share at most an edge and no two 4-cycles are incident, the only option is that $l=3$, $T_2$ starts with a 4-cycle, and there is a 4-cycle $(u_{2}u_{3}u_{4}w_{3})$ in between $T_1$ and $T_2$. Then the second cycle of $T_2$ must end in some edge $w_3v_3$, thus it must be a convex 6-cycle (see Figure \ref{fig:cases12a}).

\begin{figure}[h]

\centering
\begin{subfigure}[b]{0.3\textwidth}

\begin{tikzpicture}[style=thick,scale=0.9]
\tikzstyle{vertex}=[draw, circle, inner sep=0pt, minimum size=5.5pt]

\node (0) at (0  ,0) [vertex,label={[shift={(-0.30,-0.68)}]$v_{0}$ }] {};
\node (1) at ( 2  ,0) [vertex,label={[shift={(-0.30,-0.68)}]$u_{0}$ }] {};
\node (2) at (4  ,0) [vertex,label={[shift={(-0.30,-0.68)}]$x_{0}$ }] {};

\node (3) at ( 0  ,-1) [vertex,label={[shift={(-0.30,-0.68)}] }] {};
\node (4) at (2,-1 ) [vertex,label={[shift={(-0.30,-0.68)}]$u_{1}$ }] {};
\node (5) at ( 4,-1) [vertex,label={[shift={(-0.30,-0.68)}] }] {};

\node (6) at (0,-2) [vertex,label={[shift={(-0.30,-0.68)}] }] {};
\node (7) at ( 2,-2) [vertex,label={[shift={(-0.30,-0.35)}]$u_{2}$ }] {};
\node (8) at (4  , -2) [vertex,label={[shift={(-0.30,-0.68)}] }] {};

\node (9) at ( 0  , -3) [vertex,label={[shift={(-0.30,-0.68)}]$v_{3}$ }] {};
\node (10) at ( 1  , -3) [vertex,label={[shift={(-0.30,-0.68)}]$w_3$ }] {};
\node (11) at ( 3  , -3) [vertex,label={[shift={(-0.30,-0.45)}]$u_{3}$ }] {};
\node (12) at ( 4  , -3) [vertex,label={[shift={(-0.30,-0.68)}]$x_{3}$ }] {};

\node (13) at ( 2  , -4) [vertex,label={[shift={(-0.30,-0.68)}]$u_{4}$ }] {};

\draw (0) -- (1);
\draw (1) -- (2);
\draw (0) -- (3);
\draw (1) -- (4);
\draw (3) -- (4);
\draw (2) -- (5);
\draw (3) -- (6);
\draw (4) -- (7);
\draw (5) -- (8);
\draw (6) -- (9);
\draw (7) -- (10);
\draw (7) -- (11);
\draw (12) -- (8);
\draw (12) -- (11);
\draw (10) -- (9);
\draw (10) -- (13);
\draw (11) -- (13);
\end{tikzpicture}
\caption{Subgraph in Case (i)}
\label{fig:cases12a}
\end{subfigure}
\begin{subfigure}[b]{0.3\textwidth}

\begin{tikzpicture}[style=thick,scale=0.9]
\tikzstyle{vertex}=[draw, circle, inner sep=0pt, minimum size=5.5pt]

\node (0) at (2  ,-1) [vertex,label={[shift={(-0.30,-0.68)}]$u_{j+1}$ }] {};
\node (1) at ( 2  ,0) [vertex,label={[shift={(-0.30,-0.68)}]$u_{j}$ }] {};
\node (2) at (4  ,0) [vertex,label={[shift={(-0.30,-0.68)}]$x_{j}$ }] {};

\node (3) at ( 0  ,1) [vertex,label={[shift={(-0.30,-0.68)}] }] {};
\node (4) at (2,1 ) [vertex,label={[shift={(-0.45,-0.68)}]$u_{j-1}$ }] {};
\node (5) at ( 4,1) [vertex,label={[shift={(-0.30,-0.68)}] }] {};

\node (6) at (0,2) [vertex,label={[shift={(-0.30,-0.68)}] }] {};
\node (7) at ( 2,2) [vertex,label={[shift={(-0.50,-0.45)}]$u_{j-2}$ }] {};
\node (8) at (4  , 2) [vertex,label={[shift={(-0.30,-0.68)}] }] {};

\node (9) at ( 0  , 3) [vertex,label={[shift={(-0.45,-0.68)}]$v_{j-3}$ }] {};
\node (10) at ( 1  , 3) [vertex,label={[shift={(0.55,-0.50)}]$w_{j-3}$ }] {};
\node (11) at ( 3  , 3) [vertex,label={[shift={(0.50,-0.68)}]$u_{j-3}$ }] {};
\node (12) at ( 4  , 3) [vertex,label={[shift={(0.50,-0.68)}]$x_{j-3}$ }] {};

\node (13) at ( 2  , 4) [vertex,label={[shift={(0.50,-0.30)}]$u_{j-4}$ }] {};

\draw (0) -- (1);
\draw (1) -- (2);
\draw (0) -- (3);
\draw (1) -- (4);
\draw (3) -- (4);
\draw (2) -- (5);
\draw (3) -- (6);
\draw (4) -- (7);
\draw (5) -- (8);
\draw (6) -- (9);
\draw (7) -- (10);
\draw (7) -- (11);
\draw (12) -- (8);
\draw (12) -- (11);
\draw (10) -- (9);
\draw (10) -- (13);
\draw (11) -- (13);
\end{tikzpicture}
\caption{Subgraph in Cases (ii), (iii)}
\label{fig:cases12b}
\end{subfigure}
\caption{}
\label{fig:cases12}
\end{figure}

Second assume we have Case (ii) or (iii) and let $u_ju_{j+1}$ be the ending edge of $T_2$. Let $D_2$ be the last convex cycle of $T_2$, the one that includes $u_{j}u_{j+1}$. Assume the neighbor of $u_{j}$ on a side of $T_2$ is a vertex $w_{j-1}$ that is not on $P_1$, or in other words, we have a 4-cycle $(u_{j}u_{j-1}u_{j-2}w_{j-1})$ in between $T_1$ and $T_2$. Since $D_2$ and this 4-cycle share at most an edge, there must exist an edge $w_{j-1}z_{j-1}$ in relation $\Theta$ with $u_{j}u_{j+1}$, i.e.,~$D_2$ is a 4-cycle. Then there exist two incident 4-cycles which cannot be.

By the above, we can assume that $T_2$ ends with a convex cycle $D_2$ that includes $u_{j-1}u_{j}u_{j+1}$ in $P_1$. Let $D_1$ be the convex cycle on $T_1$ that is incident with the edge $u_{j}u_{j-1}$. Since $D_2$ and $D_1$ share at most an edge, and $u_{j-1}u_{j}u_{j+1}$ is a part of a side of $T_1$, the cycle $D_1$ must end in an edge $u_{j}x_{j}\in V(D_1)$ in relation $\Theta$ with edge $u_0x_0$. The part of the traverse $T_1$ from $u_{0}x_{0}$ to $u_jx_j$ is also traverse, say $T_1'$. Now starting from the end of traverses $T_1'$ and $T_2$ we have a similar situation as before: incident edges $u_jx_j, u_ju_{j+1}$ in which traverses end (before they started), and an edge $u_ju_{j-1}$ on the sides of $T_1'$ and $T_2$. Similar arguments as before lead us to a convex 6-cycle (see Figure \ref{fig:cases12b}).
\end{proof}

We will now consider two kinds of partial cubes: the ones that have intertwining convex cycles and the ones that do not. In both cases we will use Lemma \ref{lem:constHs} to show that $g_2(G)=6$ for a non-basic, cubic, vertex-transitive partial cube.

\begin{prop}\label{prop:nointertwine}
Let $G$ be a non-basic, cubic, vertex-transitive partial cube. If no two convex cycles share more than an edge, then $g_2(G)=6$.
\end{prop}

\begin{proof}
Let $G$ be as in the assertion of the lemma. We want to prove that the situation from Lemma \ref{lem:constHs} occurs in $G$.
Take two edges $z_0y_0$ and $z_my_m$ that are in relation $\Theta$, and assume the distance between them is maximal among all such pairs. Let $T_1$ be a convex traverse from $z_0y_0$ to $z_my_m$, provided by Lemma \ref{lem:cycles}. Let $x_1$ be the neighbor of $z_m$, different from its two neighbors on the traverse. By Lemma \ref{lem:allincident}, vertices $z_m,y_m,x_1$ lie in some convex cycle $D$. Denote the vertices of $D$ by $D=(x_0x_1x_2\ldots x_{2k-1})$, where $x_0=z_m$ and $x_{2k-1}=y_m$. We have $z_0y_0\Theta z_my_m \Theta x_{k-1}x_k$. 

By the maximality assumption, the sequence of cycles on $T_1$ together with cycle $D$ is not a traverse from $z_0y_0$ to $x_{k-1}x_k$.
On the other hand, by Lemma \ref{lem:cycles}, there is a convex traverse $T_2$ from $z_0y_0$ to $x_{k-1}x_k$. 
Denote by $P_1',P_1''$ the $z_0,z_m$-, $y_0,y_m$-side of $T_1$ and by $P_2',P_2''$ be the $z_0,x_{k-1}$-, $y_0,x_k$-side of $T_2$, respectively. The edge $z_mx_1$ lies in $D$ but does not lie in $T_1$. Assume it lies in $T_2$, say on the convex cycle $D'$ of $T_2$. As noted in the preliminaries, none of the edges with exactly one end in $D'$ is in relation $\Theta$ with any edge of $D'$ since $D'$ is convex. By definition, $D'$ has an edge in relation $\Theta$ with $z_my_m$. Thus $z_my_m$ has both its ends in $D'$, i.e.,~it lies in $D'$. Since convex cycles share at most an edge, $D'=D$ and thus $T_2$ includes $D$. Let $D''=(z_mz_{m-1}\ldots z_{m'}y_{m'}y_{m'+1}\ldots y_m)$ be the last cycle on $T_1$. Since $G$ is cubic, the traverse $T_2$ must include edge $z_mz_{m-1}$. Then also the cycle $D''$ must be in $T_2$ since convex cycles share at most an edge. Inductively we can show that $T_1\subset T_2$. A contradiction since $D\cup T_1$ is not a traverse, thus $z_mx_1$ do not lie on $T_2$.

On the other hand, $z_mx_1$ lies on a closed walk starting in $z_0$ passing a side of $D$, $P_1'$ and $P_2'$. Since in partial cubes every such closed walk must pass the $\Theta$-class of $z_mx_1$ at least twice, there must be another edge on it in relation $\Theta$ with $z_mx_1$. Since $D$ is convex, it must be on $P_1'$ or $P_2'$.

If there is an edge in relation $\Theta$ with $z_mx_1$ on $P_1'$, then Lemma \ref{lem:paste2}, provides a convex traverse $T_3$ attached to $P_1'$ as in Cases (i)-(iii). By Lemma \ref{lem:constHs} and the fact that convex cycles in $G$ share at most an edge, we have $g_2(G)=6$. If there is an edge in relation $\Theta$ with $z_mx_1$ on $P_2'$, then there also exists an edge in relation $\Theta$ with $z_mx_1$ on $P_2''$. The latter holds since paths $P_2'$ and $P_2''$ cross the same $\Theta$-classes. Considering this edge on $P_2''$ and the edge $x_{k}x_{k+1}$, with $x_{k}x_{k+1}\Theta z_mx_1$, Lemma \ref{lem:paste2} gives a convex traverse $T_4$ attached to $P_2''$ as in Cases (i)-(iii). By Lemma \ref{lem:constHs}, it holds that $g_2(G)=6$. This concludes the proof.
\end{proof}

\begin{prop}\label{prop:yesintertwine}
If $G$ is a non-basic, cubic, vertex-transitive partial cube with two convex cycles that share more than one edge, then $g_2(G)=6$.
\end{prop}

\begin{proof}
Assume that two convex cycles in $G$ share more than an edge. 
By Claim \ref{clm:convex_always_intert}, the two convex cycles intertwine.
Denote the vertices of the intertwining cycles by $(v_0v_1\ldots v_mv_{m+1}\ldots v_{2m+2n_1-1})$ and $(u_0u_1\ldots u_mu_{m+1}\ldots u_{2m+2n_2-1})$  where $u_0=v_0,\ldots ,u_m=v_m$. By definition $m\geq 2$; moreover assume that the residue of intertwining $n_1+n_2$ is minimal among all pairs of intertwining convex cycles. 

It holds that $v_{2m+n_1}v_{2m+n_1-1} \Theta v_{m-1}v_m \Theta u_{2m+n_2}u_{2m+n_2-1}$ and that $v_{2m+n_1-1}v_{2m+n_1-2} \Theta v_{m-2}v_{m-1} \Theta u_{2m+n_2-1}u_{2m+n_2-2}$. By Lemma \ref{lem:cycles}, we have a traverse $T_1$ from  $v_{2m+n_1-1}v_{2m+n_1}$ to $u_{2m+n_2-1}u_{2m+n_2}$. Let $P_1$ be the $v_{2m+n_1-1},u_{2m+n_2-1}$-side of it.
We denote by  $D_1,\ldots,D_i$ the convex cycles on it, and let vertices in $P_1$ be denoted by $P_1=z_0z_1\ldots z_{k}$, where $v_{2m+n_1-1}=z_0$ and $u_{2m+n_2-1}=z_k$. Notice that $v_{2m+n_1}v_{2m+n_1+1}\ldots v_{2m+2n_1-1}v_{0}u_{2m+2n_2-1}\ldots u_{2m+n_2}$ is a $v_{2m+n_1} , u_{2m+n_2}$-path of length $n_1+n_2$, thus the length of the traverse $T_1$ is at most $n_1+n_2$.

Consider the isometric path $P_1$ and edges $v_{2m+n_1-1}v_{2m+n_1-2}$ and $u_{2m+n_2-1}u_{2m+n_2-2}$ that are in relation $\Theta$. By Lemma \ref{lem:paste2}, we have a traverse $T_2$ attached to $P_1$ as in one of the Cases (i)-(iii).
To prove that the situation from Lemma \ref{lem:constHs} occurs in $G$ and thus that $g_2(G)=6$, we have to prove that the convex cycles of $T_1$ and the convex cycles of $T_2$ pairwise intersect in at most an edge.

Let $E_1,\ldots, E_{\tilde{i}}$ be convex cycles that form the traverse $T_2$ from $v_{2m+n_1-1}v_{2m+n_1-2}$ to $u_{2m+n_2-1}u_{2m+n_2-2}$ in Case (i),  from $z_{i_1}y_{i_1}$ to $z_{i_2}z_{i_2+1}$ in Case (ii), or from $w_{i_1}y_{i_1}$ to $z_{i_2}z_{i_2+1}$ in Case (iii), for some $0\leq i_1<i_2< k ,w_{i_1},y_{i_1}$.
We want to prove that no two cycles from $D_1,\ldots, D_i,E_1,\ldots,E_{\tilde{i}}$ intersect in more than an edge. For the sake of contradiction, assume $E_{\tilde{l}}$ and $D_{\tilde{k}}$ share at least two edges. By Claim \ref{clm:convex_always_intert} they intertwine. We will get to a contradiction by showing that $E_{\tilde{l}}$ and $D_{\tilde{k}}$ have the residue of intertwining smaller than $n_1+n_2$.

We have multiple options. Firstly, assume that $E_{\tilde{l}}$ intersects with exactly one of the 4-cycles in between the traverses, say with $(z_{l'-1}z_{l'}z_{l'+1}w_{l'})$ for some $0<l'<k, w_{l'}$. Since $G$ is cubic $z_{l'-1}\neq z_0 =v_{2m+n_1-1}$, thus $l'\geq 2$. By Claim \ref{clm:4cycle_no_intert}, $E_{\tilde{l}}$ shares exactly an edge with the 4-cycle. Without loss of generality we can assume a side of $E_{\tilde{l}}$ is of a form $w_{l'}z_{l'+1}\ldots z_{l''}$, for some $2\leq l'<l''\leq n_1+n_2$, where $w_{l'}z_{l'+1}$ is an edge of the 4-cycle in between $T_1$ and $T_2$. A side of $D_{\tilde{k}}$ is of a form $z_{k'}z_{k'+1}\ldots z_{k''}$, for some $0\leq k'<k''\leq n_1+n_2$. Since $D_{\tilde{k}}$ does not intersect in more than an edge with the 4-cycle $(z_{l'-1}z_{l'}z_{l'+1}w_{l'})$, we have $2\leq l' \leq k'$. Moreover, because $E_{l'}$ and $D_{k'}$ intersect in at least two edges, it holds that $k''-l'\geq 3$, and thus $k''\geq 5$. 

Assume $l'=k'$ and $k''\leq l''$. Recall that the residue of intertwining is calculated as $i(E_{\tilde{l}},D_{\tilde{k}})=(d_1+d_2-4d_3)/2$, where $d_1$ is the length of $E_{\tilde{l}}$, $d_2$ the length of $D_{\tilde{k}}$ and $d_3$ the number of edges they share.  In particular, for $E_{\tilde{l}}$ and $D_{\tilde{k}}$ we get:
\begin{align*}
i(E_{\tilde{l}},D_{\tilde{k}}) & =(2(l''-l')+2+2(k''-k')+2-4(k''-(k'+1)))/2 \\
& =l''-k''-l'+k'+4\leq (n_1+n_2)-5+4<n_1+n_2.
\end{align*}
A contradiction. We get similar results in all the other cases, that is, if $l'=k'$, but $l''\leq k''$, and if $l' < k'$ and $k''\leq l''$ or $k''\geq l''$.

Now assume $E_{\tilde{l}}$ intersects with exactly two of the 4-cycles in between the traverses, i.e.,~a part of $E_{\tilde{l}}$ is of the form $w_{l'}z_{l'+1}\ldots z_{l''-1},w_{l''}$, where $w_{l'}z_{l'+1}$ and $z_{l''-1}w_{l''}$ are edges of the 4-cycles in between the traverses. Then, since a side of $D_{\tilde{k}}$ is of the form $z_{k'}z_{k'+1} \ldots z_{k''}$, we can, similarly as above, show that $D_{\tilde{k}}$ and $E_{\tilde{l}}$ intertwine with the residue of intertwining smaller than $n_1+n_2$.

Since $E_{\tilde{l}}$ does not intersect in more than an edge with 4-cycles in between the traverses, it cannot be incident with more than two of them. Therefore, we can assume it does not intersect with any of them. In this case a side of $E_{\tilde{l}}$ is of the form $z_{l'}z_{l'+1} \ldots z_{l''}$, while a side of $D_{\tilde{k}}$ is of the form $z_{k'}z_{k'+1} \ldots z_{k''}$. Assume $l'\leq k'<l''\leq k''$, i.e.,~they share the path between $z_{k'}$ and $z_{l''}$ (with $l''-k'\geq 2$ since they intertwine). We have
$$i(E_{\tilde{l}},D_{\tilde{k}})=(2(l''-l')+2+2(k''-k')+2-4(l''-k'))/2=k''-l'-2(l''-k'-2)\leq k''-l'\leq n_1+n_2.$$
Since the equality must hold, we have $l''-k'=2$. We know that $G$ must have girth 4. Thus the vertex $z_{k'+1}$ must be incident with a 4-cycle $D_4$. Since $G$ is cubic, this 4 cycle must share two consecutive edges with  $D_{\tilde{k}}$ or $E_{\tilde{l}}$, a contradiction with Claim \ref{clm:4cycle_no_intert}. We get similar outcomes with the other positions of $l',k',l'', k''$. This proves that cycles of $T_1$ and $T_2$ pairwise share at  most an edge, concluding the proof.
\end{proof}

The following proposition, combined with Propositions \ref{prop:nointertwine}, \ref{prop:yesintertwine}, Corollary \ref{cor:basic}, and Lemma \ref{lem:6-cycle} proves Theorem \ref{thm:thethm}. To prove it we shall need a result from group theory. A group of the form $\langle \alpha_1,\ldots, \alpha_n\, |\, (\alpha_i\alpha_k)^{k_{ij}}=1 \rangle$, where $k_{ii}=1,k_{ij}=k_{ji}>1$ is called a \emph{Coxeter group}. The finite Coxeter groups are classified with their Cayley graphs being partial cubes due to their connection with oriented matroids and reflection arrangements \cite[Chapter 2.3]{matroid_bj}. Limiting ourselves to those Coxeter groups whose Cayley graphs are cubic, we are left with the following graphs: cubic permutahedron, the truncated cuboctahedron, and the truncated icosidodecahedron.

\begin{prop}\label{lem:spec}
If $G$ is a non-basic, cubic, vertex-transitive partial cube with $g_2(G)=6$, then  $G$ is isomorphic to the truncated cuboctahedron or the truncated icosidodecahedron.
\end{prop}

\begin{proof}
Let $G$ be a non-basic, cubic, vertex-transitive partial cube, with $g_2(G)=6$. Then $g_1(G)=4$ and $g_3(G)=k>6$, by Corollary \ref{cor:basic} and Lemma \ref{lem:6-cycle}. It follows that the vertex-stabilizers of the automorphism group of $G$ are trivial. Thus $G$ is a Cayley graph. Notice that by Claim \ref{clm:4cycle_no_intert} each 4-cycle shares at most one edge with any convex 6- and $k$-cycle. Moreover, if a convex 6-cycle $D_6$ shares more than an edge with a convex $k$-cycle $D_k$, then they  intertwine by Claim \ref{clm:convex_always_intert}. Since both cycles are convex, they can share at most two consecutive edges, say $v_1v$ and $vv_2$. If this is the case, let $D_4$ be the convex 4-cycle incident with $v$. Then $D_4$ must share two edges with $D_6$ or $D_k$ since $G$ is cubic. A contradiction, thus also convex 6-cycles and convex $k$-cycles share at most an edge. 

Color the edges that simultaneously lie in a 4-cycle and a convex 6-cycle green, the edges in a 4-cycle and a convex $k$-cycle red, and the edges in a convex 6-cycle and a convex $k$-cycle blue. By the above discussion the colors are well defined. We get the relations
$\langle r, g, b\, |\, r^2, g^2, b^2, (rg)^2, (gb)^3, (br)^{k/2}\rangle.$
Then this is a Coxeter group. By the classification of finite Coxeter groups, and the fact that $k>6$, the only possibility is that $k$ equals 8 or 10. Moreover, in the first case $G$ is isomorphic to the truncated cuboctahedron, while in the second it is isomorphic to the truncated icosidodecahedron.
\end{proof}

\section{Concluding remarks}

With this article we have provided a classification of cubic, vertex-transitive partial cubes. Since the variety  of such graphs is rather small, it suggests that a similar classification can be done for graphs with higher valencies. The latter problem is wide open. The regular graphs in the subcubic cases can be seen as the beginnings of greater families of vertex-transitive partial cubes: $G(10,3)$ as a middle level graph;  the cubic permutahedron, the truncated cuboctahedron, the truncated icosidodecahedron, even cycles, and $K_2$ as Cayley graphs of finite Coxeter groups; and even prisms as the Cartesian products of the latter graphs.
To our knowledge these families are the only known examples of vertex-transitive partial cubes. We do not assume that these are the only ones, but we would like to propose the following conjecture:

\begin{conj}
The middle level graphs are the only vertex-transitive partial cubes with girth six.
\end{conj}

This paper was motivated by a computer search for partial cubes on a census of cubic, vertex-transitive graphs up to 1280 vertices \cite{potovcnik2013cubic}. For the search we transferred the basis in Sage environment \cite{sage} and used Eppstein's algorithm \cite{eppstein2011recognizing} for the recognition of partial cubes. We are thankful to the authors of the census and the algorithms. We would also like to thank Sandi Klavžar and the reviewers for useful comments on the text.

\bibliography{biblio}

\end{document}